\documentclass[pre,amsmath,amssymb,aps,superscriptaddress,flushbottom]{revtex4-2}

\usepackage[utf8]{inputenc}
\usepackage[T1]{fontenc}
\usepackage{graphicx}
\usepackage{float}
\usepackage{color}
\usepackage{xr}
\usepackage{bm}
\usepackage{amsmath}
\usepackage{amsthm}
\usepackage{enumerate}
\usepackage{enumitem}
\usepackage{hhline}
\usepackage{subcaption}
\usepackage{mathrsfs}
\newtheorem{theorem}{Theorem}
\newtheorem{lemma}[theorem]{Lemma}
\theoremstyle{definition}
\newtheorem{definition}{Definition}[section]
\newcommand\myeq{\mathrel{\overset{\makebox[0pt]{\mbox{\normalfont\tiny\sffamily fit}}}{=}}}
\begin{document}
\title{Wave Kinetics and Thermalization in Kadomtsev-Petviashvili-I System}
\author{Kolluru Venkata Kiran} 
\email{kiran8147@gmail.com}
\affiliation{Universit\`e C\^ote d’Azur, CNRS-Institut de Physique de Nice, 17 Rue Laupr\^etre, Nice, 06200, France}
\author{Giorgio Krstulovic} 
\email{giorgio.krstulovic@univ-cotedazur.fr}
\affiliation{Universit\`e C\^ote d’Azur, CNRS-Institut de Physique de Nice, 17 Rue Laupr\^etre, Nice, 06200, France}

\author{Alan.C.Newell} 
\email{anewell@arizona.edu}
\affiliation{Mathematics Department, University of Arizona,, 617 N Santa Rita Ave , AZ, Tucson, 85721, United States}
\author{Sergey Nazarenko} 
\email{sergey.nazarenko@univ-cotedazur.fr}
\affiliation{Universit\`e C\^ote d’Azur, CNRS-Institut de Physique de Nice, 17 Rue Laupr\^etre, Nice, 06200, France}

\begin{abstract}{
}
We study properties of solutions, both evolving and equilibrium of the wave-kinetic equation describing ensembles of weak random waves governed by the Kadomstev-Petviashivli-I equations. The latter equation is integrable by the inverse scattering method, and yet it allows resonant wave interactions leading to redistribution of energy in the Fourier space. Such resonant interactions preserve an infinite number of invariants and we find that they preserve compactness of Fourier space supports. Numerically, we observe that the system can 
thermalize to one of the equilibrium states of Rayleigh-Jeans type, despite the common empirical belief that thermalization is impossible for integrable systems. The thermalized states are formed via non-local spectral transfers leading to creation of strong low-wavenumber peaks of the wave spectrum -- a process akin to Bose-Einstein condensation.
\end{abstract}
\maketitle
\section{Introduction}
Integrability and thermalization are generally perceived as two contrasting paradigms for understanding nonlinear dynamics. Thermalization describes the tendency of a system to evolve toward a state in which its macroscopic properties are captured by equilibrium statistical ensembles, such as the microcanonical or canonical ensembles. In classical systems, nonlinear interactions and chaotic dynamics promote mixing in relevant phase space, thereby facilitating thermalization.
In contrast, integrable systems exhibit quasi-periodic and non-ergodic behavior due to the presence of additional conserved quantities, equal in number to the system’s degrees of freedom. This distinction is well illustrated by the classical and historically significant \emph{Fermi–Pasta–Ulam–Tsingou} (FPUT) chain. The model consists of a one-dimensional array of identical masses connected by nonlinear springs, interacting only with their nearest neighbors. It was originally devised by Fermi and collaborators to investigate a question raised three decades earlier by Debye. Debye had argued that, since heat is a form of vibrational motion, it should propagate with infinite conductivity, much like a wave, rather than exhibiting finite conductivity. He further proposed that incorporating nonlinearities might resolve this issue by inducing mode mixing and thereby leading to thermalization, with the associated relaxation time providing a measure of finite conductivity. Contrary to the initial expectation that nonlinear interactions would drive the system toward thermal equilibrium from arbitrary initial conditions, numerical experiments revealed a far more intricate dynamical behavior characterized by quasi-periodicity. Moreover, any relaxation toward thermal equilibrium—if it occurs at all—takes place on timescales far longer than those suggested by the inverse strength of the nonlinearity~\cite{fermi1955studies,newell1985solitons,berman2005fermi,gallavotti2007fermi,dauxois2008fermi}.

Motivated by these findings, Zabusky and Kruskal~\cite{zabusky1965interaction} revisited the FPUT problem from a continuum perspective, deriving a nonlinear dispersive partial differential equation (PDE) that can be mapped, in an appropriate asymptotic regime, to the Korteweg–de Vries (KdV) equation. They showed that the KdV equation admits soliton solutions—spatially localized, stable traveling waves whose interactions are elastic and result only in phase shifts, offering a plausible justification for the quasiperiodic behaviour of FPUT chains.  Consequently, Gardner, Greene, Kruskal, and Miura introduced the Inverse Scattering Transform (IST) method~\cite{gardner1967method}, which recasts the nonlinear evolution into a linear spectral problem for an associated Lax operator. Within this framework, they showed that the solitons are exact closed form solutions of the KdV equations asymptotically appearing from a broad class of initial conditions. 
%\textcolor{red}{It is important to note, however, that the integrability of the system holds only within the asymptotic regime considered by Kruskal and collaborators. At higher orders in amplitude, the effective continuum description of the FPUT chain develops nontrivial resonances—appearing, for instance, at sixth order—which enable energy transfer and give rise to a finite flux spectrum~\cite{onorato2015route, lvov2018double}.}\\ 

%\Secondpa
The IST formalism was subsequently extended to a large class of nonlinear dispersive PDEs. Hallmarks of integrability, such as the existence of infinitely many conserved quantities and soliton solutions, led to the widespread perception that integrable systems cannot thermalize and, more generally, do not warrant a statistical description comparable to that of chaotic nonlinear systems.  This raises a fundamental question: can integrable nonlinear dispersive PDEs exhibit sufficiently complex behavior to warrant a statistical description, and if so, what is the appropriate theoretical framework?

The answer to the latter question is more accessible. The natural statistical framework for describing the emergent, effectively chaotic behavior of nonlinear dispersive PDEs is wave turbulence theory (WTT)~\cite{zakharov2012kolmogorov,nazarenko2011wave,newell2011wave}. In this approach, the physical system is viewed as an ensemble of weakly interacting dispersive waves. 
Under the mild assumption that spatial correlations in the initial state decay sufficiently rapidly, the phases of the Fourier amplitudes can be treated as effectively random. This leads to a natural asymptotic closure for the hierarchy of Fourier-space cumulants, allowing one to derive a closed wave-kinetic equation (WKE) governing the evolution of the spectral number density. The higher-order cumulants contribute corrections that can be interpreted as renormalizations of the wave frequencies.\\
%For  wave fields with short correlations in the physical space (consequently, random initial wave phases and amplitudes),     one can derive a wave-kinetic equation (WKE) governing the evolution of wave action spectrum.
%
Later, Zakharov~\cite{zakharov2009turbulence} conjectured that integrable nonlinear dispersive PDEs can be divided into two distinct classes: strongly integrable and weakly integrable systems. While both classes are integrable via the IST, strongly integrable systems possess a \emph{complete} set of invariants, rendering a kinetic description irrelevant. In contrast, weakly integrable systems, despite also admitting infinitely many conserved quantities, do not have a complete set, thereby allowing for the derivation of a nontrivial WKE.
\begin{figure}[!]
\centering
\includegraphics[width=
6cm
%\linewidth
]{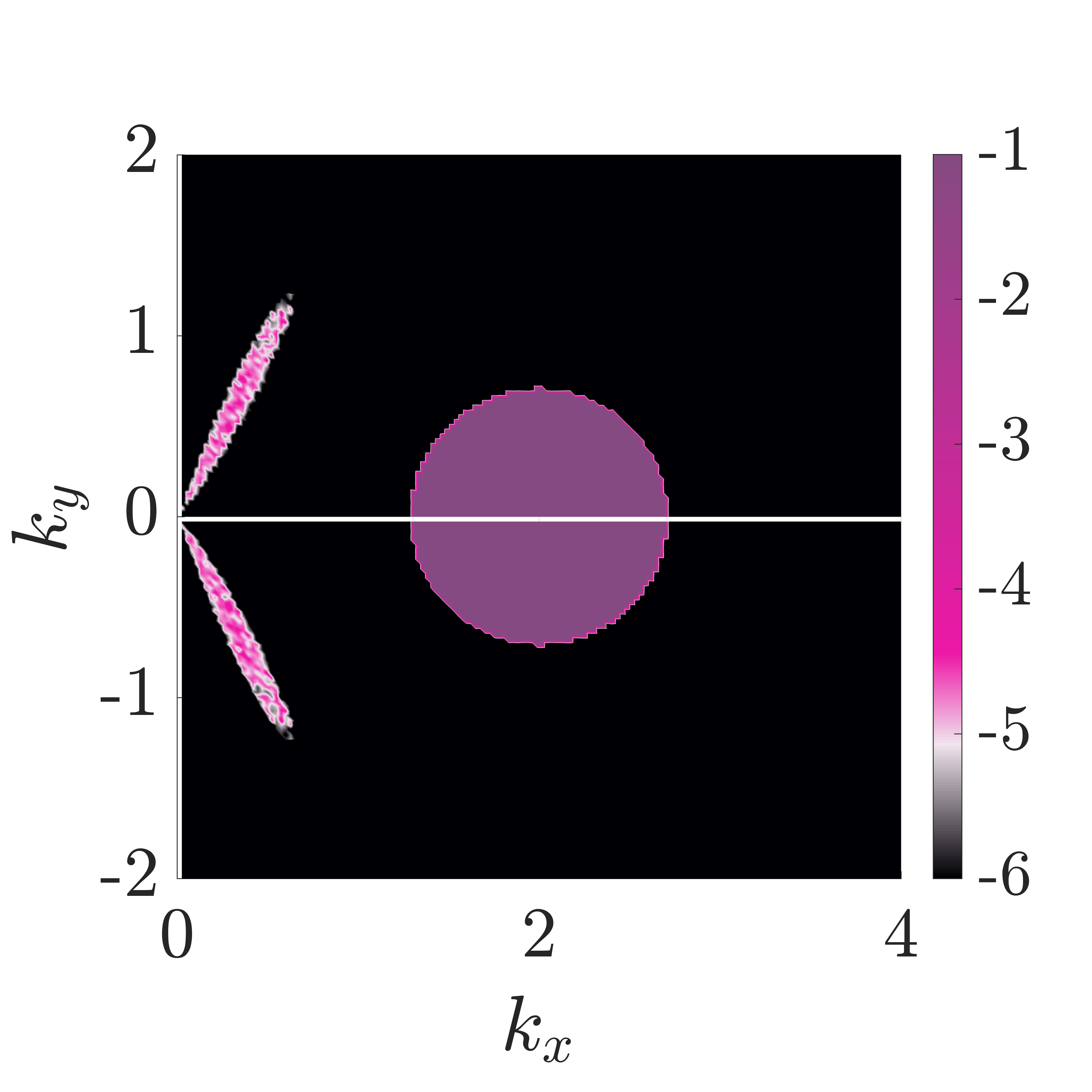}
\caption{Logarithmic color plot of the anisotropic Fourier amplitude spectrum $\vert \hat{u}(\bm k)\vert$ after $\sim\mathcal{O}(10^3)$ linear timescales corresponding to the mode $\bm k=(2,0)$.  }
\label{Fig:dns1}
\end{figure}
A prototypical example distinguishing these two classes is provided by the Kadomtsev–Petviashvili (KP) equations, which describe the evolution of a two-dimensional scalar wave field $u(x,y,t)$~\cite{kadomtsev1970stability,ablowitz2011nonlinear}. In nondimensional form, they are given by
\begin{equation}\label{Eq:KP}
\partial_x\big(\partial_t u + u_{xxx} + 6u u_x\big) = -3\sigma u_{yy}.
\end{equation}
Here, $\sigma = \pm 1$ distinguishes between the two systems: $\sigma = -1$ corresponds to the KP-II equation, while $\sigma = 1$ corresponds to the KP-I equation. Both arise as asymptotic models of weakly nonlinear, weakly dispersive water waves in the long-wavelength limit~\cite{Ablowitz_Segur_1979}. The KP-II equation is relevant in the regime dominated by gravity, whereas the KP-I equation models waves in the presence of strong surface tension. KP-I also models strongly anisotropic acoustic waves in 2D Bose-Einstein condensates and in nonlinear optics~\cite{KuznetsovTuristyn1988}.
From the perspective of integrability, KP-II is classified as strongly integrable, while KP-I falls into the class of weakly integrable systems. As a consequence, the KP-I equation admits a nontrivial WKE, whereas such a description is absent in the KP-II case.

In this work, we demonstrate that nonlinear dispersive PDEs that are integrable via the IST can nevertheless exhibit Fourier-space fluxes leading to thermalization in the weakly nonlinear, wave-kinetic regime—whenever such a regime exists. Our analysis focuses on the KP-I equation, which despite it's integrability, allows for Fourier space fluxes. As a quick motivational illustration of this fact, in Fig.~\ref{Fig:dns1} we plot the absolute value of the anisotropic Fourier amplitudes, 
defined in Eq.~\eqref{Eq:def_FT},
%$\hat{u}(\mathbf{k},t)=\frac{1}{L^{2}}\int_{[0,L]^2} u(\mathbf{x},t)\,e^{-i\mathbf{k}\cdot\mathbf{x}}\,d\mathbf{x}$ 
obtained by direct numerical simulation (DNS) of Eq.~\eqref{Eq:KP} starting with an initial condition supported in a $k$-space circle having constant  amplitudes and random phases   (see Appendix.~\ref{methods:DNS_KP} for more details.). The results in Fig.~\ref{Fig:dns1} reveal clear evidence of nonlocal energy transfer across Fourier space into two symmetric low-$k$ beam-like spectra. In the present paper, we will demonstrate that these features are naturally captured within the framework of wave turbulence theory, leading to a WKE that exhibits thermalization.
Remarkably, we further show analytically that the resulting thermal state is physically meaningful: it undergoes self-truncation in Fourier space, thereby confining energy to a finite spectral domain and avoiding the ultraviolet catastrophe typically associated with equilibrium statistical descriptions.\\

%We emphasize in this work we consider WTT that arises exactly for integrable systems, as opposed to WTT appearing as a result of taking into account higher corrections that break the integrability. For example, the effective continuum description of the FPUT chain develops nontrivial resonances—appearing, for instance, at sixth order—which enable energy transfer and give rise to a finite flux spectrum~\cite{onorato2015route, lvov2018double}, similary for the 1D nonlinear schrodinger equation ~\cite{laurie2012one,bortolozzo2009optical}. Furthermore, our results are also distinct from the exact WTT description of 1D nonlinear schrodinger equation in which the kinetic interactions vanish asympotically in time~\cite{suret2011wave}.

We emphasize that, in this work, we consider WTT for an {\em exactly} integrable system. This is distinct from results where nontrivial WTT emerges as a consequence of higher-order corrections that {\em break} integrability. For example, the effective continuum description of the FPUT  chain with the next to the KdV order included develops nontrivial  sixth-order resonances which enable energy transfers and give rise to a finite flux spectrum~\cite{onorato2015route, lvov2018double}. A similar situation occurs in one-dimensional (1D) nonlinear optics where corrections to the Nonlinear Schrödinger (NLS) equation are taken into account ~\cite{laurie2012one, bortolozzo2009optical}. Furthermore, WTT considered here is also different from a generalized WTT description of the integrable 1D NLS equation, where kinetic interactions vanish asymptotically in time~\cite{suret2011wave}.

\section{Wave turbulence theory}
In this section, we briefly review the WTT for the KP-I equation  \cite{zakharov2009turbulence}. We define the Fourier transform of $u(\mathbf{x},t)$ on a periodic domain $[0,L]^2$ as
\begin{equation}
\hat{u}(\mathbf{k},t)=\frac{1}{L^{2}}\int_{[0,L]^2} u(\mathbf{x},t)\,e^{-i\mathbf{k}\cdot\mathbf{x}}\,d\mathbf{x},
\label{Eq:def_FT}
\end{equation}
where $\mathbf{k}=(k_x,k_y)$. In the absence of nonlinearity, $\hat{u}(\mathbf{k},t)$ rotates in the complex plane with the angular  velocity given by the dispersion relation
\begin{equation}\label{Eq:disp_kxky}
\omega(\mathbf{k})=k_x^{3}+3\frac{k_y^{2}}{k_x}.
\end{equation}

In the weakly nonlinear regime, the KP-I equation admits a statistical description within WTT~\cite{zakharov2009turbulence}. Assuming random initial phases and amplitudes, and taking first the large-box limit ($L\to\infty$) followed by the weak-nonlinearity limit, one obtains a WKE for the wave-action spectrum
\begin{equation}
n_{\mathbf{k}} \equiv \left(\frac{L}{2\pi}\right)^{2}\frac{\langle|\hat{u}(\mathbf{k})|^{2}\rangle}{k_x},
\end{equation}
where $\langle\cdot\rangle$ denotes averaging over initial phases and amplitudes. The WKE reads
\begin{equation}\label{eq:wke}
\frac{\partial n_{\mathbf{k}}}{\partial t}
=
\int_{k_{1x}\ge 0} \!\!\!\!\!\!d\mathbf{k}_1 \int_{k_{2x}\ge 0} \!\!\!d\mathbf{k}_2
\left(\mathscr{R}_{{\mathbf{k}_1} {\mathbf{k}_2} \mathbf{k}}-\mathscr{R}_{\mathbf{k}{\mathbf{k}_1} {\mathbf{k}_2}}-\mathscr{R}_{{\mathbf{k}_2} \mathbf{k}{\mathbf{k}_1}}\right),
\end{equation}
with
\begin{eqnarray}
    \label{Eq:r12k}
\mathscr{R}_{{\mathbf{k}_1} {\mathbf{k}_2} \mathbf{k}}
=
2\pi\left|V^{\mathbf{k}}_{{\mathbf{k}_1} {\mathbf{k}_2}}\right|^{2}
\delta(\mathbf{k}-\mathbf{k}_1-\mathbf{k}_2)\,
\delta(\omega_{\mathbf{k}}-\omega_{\mathbf{k}_1}-\omega_{\mathbf{k}_2}) 
\nonumber \\ 
\times \left(n_{\mathbf{k}_1}n_{\mathbf{k}_2}-n_{\mathbf{k}}n_{\mathbf{k}_1}-n_{\mathbf{k}}n_{\mathbf{k}_2}\right).
\end{eqnarray}
The delta functions enforce resonant triads satisfying
\begin{equation}\label{Eq:rmanifold}
\mathbf{k}=\mathbf{k}_1+\mathbf{k}_2,\quad
\omega_{\mathbf{k}}=\omega_{\mathbf{k}_1}+\omega_{\mathbf{k}_2},
\end{equation}
and permutations thereof. The interaction coefficient is $V^{\mathbf{k}}_{12}=\sqrt{k_x k_{1x} k_{2x}}$. Since $u$ is real-valued, it suffices to consider modes with $k_x\geq 0$. Owing to the resonant conditions in Eq.~\eqref{Eq:rmanifold}, the WKE~(\ref{eq:wke}) admits a stationary solution  $n(k_x,k_y)\sim 1/\omega(\mathbf{k})$ which corresponds to a thermal equipartition of  energy among the wave modes, and which is known as the Rayleigh-Jeans (RJ) spectrum. As we will shortly see, it represents only one particular solution from  an infinite family of solutions of the RJ type.\\
 Equation.~\eqref{eq:wke} has infinitely many conserved quantities, all linear in $n_{\mathbf{k}}$ (and hence quadratic in the wave amplitude)~\cite{zakharov2009turbulence}. In general, an invariant can be written as
\begin{equation}
I = \int_{k_x>0} \lambda_{\mathbf{k}}\, n_{\mathbf{k}}\, d\mathbf{k},
\end{equation}
where $\lambda_{\mathbf{k}}$ is the corresponding spectral density. 
Suppose that the integral defining $I$ converges. Then $I$ is  an invariant of the WKE if its density $\lambda_{\mathbf{k}}\equiv\lambda(k_x,k_y)$  satisfies
\begin{equation}\label{Eq:invcond}
    \lambda_{\mathbf{k}_1} + \lambda_{\mathbf{k}_2} - \lambda_{\mathbf{k}} = 0,
\end{equation}
on the resonant manifold defined by Eq.~\eqref{Eq:rmanifold}; see Refs.~\cite{zakharov1980degenerative,zakharov1988additional}. 
%Thus, $I$ is an invariant if and only if \eqref{Eq:invcond} defines the same resonant manifold as \eqref{Eq:rmanifold}.
It follows immediately that the choices $\lambda_{\mathbf{k}}=k_x$, $\lambda_{\mathbf{k}}=k_y$, and $\lambda_{\mathbf{k}}=\omega_{\mathbf{k}}$ yield conserved quantities. More generally,  degeneracy of the resonant manifold calls into  existence infinitely many of extra invariants~\cite{zakharov1980degenerative,zakharov1988additional}.

To make this explicit for KP-I, we introduce the following parametrization~\cite{zakharov1980degenerative,zakharov1988additional},
\begin{equation}\label{Eq:parametrization}
    k_x=\xi-\eta, \qquad k_y=\xi^{2}-\eta^2,
\end{equation}
with $\xi>\eta$ (since $k_x>0$). The dispersion relation becomes $\omega(\xi,\eta)=4(\xi^{3}-\eta^{3})$, and the resonant manifold reduces to
\begin{align}
    \xi_1-\eta_1+\xi_2-\eta_2 &= \xi-\eta, \nonumber \\
    \xi_1^{2}-\eta_1^{2}+\xi_2^{2}-\eta_2^{2} &= \xi^{2}-\eta^{2}, \\
    \xi_1^{3}-\eta_1^{3}+\xi_2^{3}-\eta_2^{3} &= \xi^{3}-\eta^{3}\nonumber.
\end{align}
These admit only two nontrivial solutions:
\begin{align}\label{Eq:rmsolutions}
    &\xi_1=\eta_2,\quad \xi_2=\xi,\quad \eta_1=\eta, \\
    &\xi_2=\eta_1,\quad \xi_1=\xi,\quad \eta_2=\eta.
    \label{Eq:rmsolutions1}
\end{align}

It then follows that, for any function $F$, the combination $F(\xi)-F(\eta)$ satisfies
\[
F(\xi)-F(\eta)=F(\xi_1)-F(\eta_1)+F(\xi_2)-F(\eta_2),
\]
on the resonant manifold subject to solutions~\eqref{Eq:rmsolutions} or~\eqref{Eq:rmsolutions1}. Therefore,
\begin{equation}
    I=\int_{-\infty}^{\infty}\int_{\eta}^{\infty} d\xi\, d\eta \; 2(\xi-\eta)\big[F(\xi)-F(\eta)\big]\, n(\xi,\eta)
    \label{Eq:inv}
\end{equation}
is conserved, where $2(\xi-\eta)$ is the Jacobian associated with the transformation in Eq.~\eqref{Eq:parametrization}.

In $\xi,\eta$ variables, the WKE becomes:
%~\eqref{Eq:wke_xieta}
%\begin{widetext}
%\begin{figure*}[!b]
\begin{widetext}
\begin{align}\label{Eq:wke_xieta}
    \frac{\partial n(\xi,\eta)}{\partial t}
    &=\int_{\eta}^{\xi} (\xi-\eta)(\lambda-\eta)^2(\xi-\lambda)^2
    \big[n(\lambda,\eta)n(\xi,\lambda)-n(\xi,\eta)n(\lambda,\eta)-n(\xi,\eta)n(\xi,\lambda)\big] d \lambda \nonumber \\
    &+\int_{-\infty}^{\eta} (\xi-\eta)(\eta-\lambda)^2(\xi-\lambda)^2
    \big[n(\eta,\lambda)n(\xi,\lambda)-n(\xi,\eta)n(\eta,\lambda)-n(\xi,\eta)n(\xi,\lambda)\big] d \lambda \nonumber \\
    &+\int_{\xi}^{\infty} (\xi-\eta)(\lambda-\eta)^2(\lambda-\xi)^2
    \big[n(\lambda,\eta)n(\lambda,\xi)-n(\xi,\eta)n(\lambda,\eta)-n(\xi,\eta)n(\lambda,\xi)\big] d \lambda.
\end{align}
\end{widetext}
%\end{figure*}

The WKE admits an infinite family of generalized Rayleigh-Jeans (RJ) solutions of the form~\cite{zakharov2009turbulence}:
\begin{equation}{\label{Eq:rjsol}}
    n_{RJ}(\xi,\eta)= \frac{1}{F(\xi)-F(\eta)}.
\end{equation}
 This solution corresponds to a thermal equilibrium state realizing a $k$-space equipartition of the invariant $I$ with density $F(\xi)-F(\eta)$.

The WKE satisfies an H-theorem~\cite{zakharov2012kolmogorov,nazarenko2011wave}, implying that the entropy $\mathcal{S}=\int 2(\xi-\eta) \log n(\xi,\eta)~ d\xi d\eta$ is maximized by the Rayleigh-Jeans (RJ) distribution. Thus, for arbitrary initial conditions Eq.~\eqref{Eq:wke_xieta} (or equivalently Eq.~\eqref{eq:wke}) evolves toward thermal equilibrium. This raises several interesting questions: firstly, how the infinitely many conserved quantities constrain the dynamics of WKE and, secondly, whether a sufficiently general initial condition evolves to an RJ state and, if so, what form of $n_{RJ}$ or $F(x)$ can one expect upon thermalization.

\begin{figure*}[t]
\centering
\begin{subfigure}{0.34\textwidth}
    \centering
    \includegraphics[width=\linewidth]{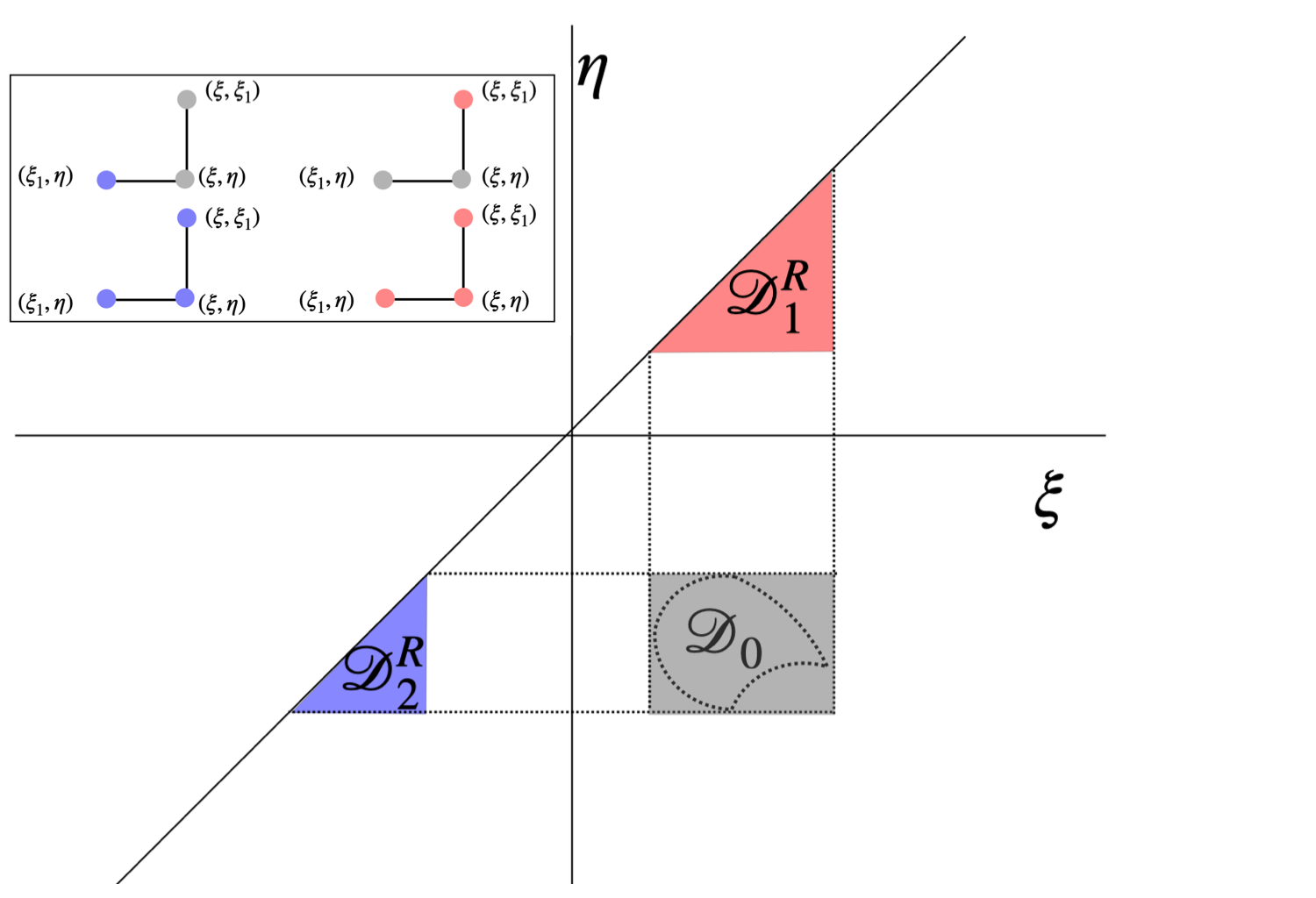}
    \caption{}
    \label{fig1:subfig1}
\end{subfigure}
\hfill
\begin{subfigure}{0.32\textwidth}
    \centering
    \includegraphics[width=\linewidth]{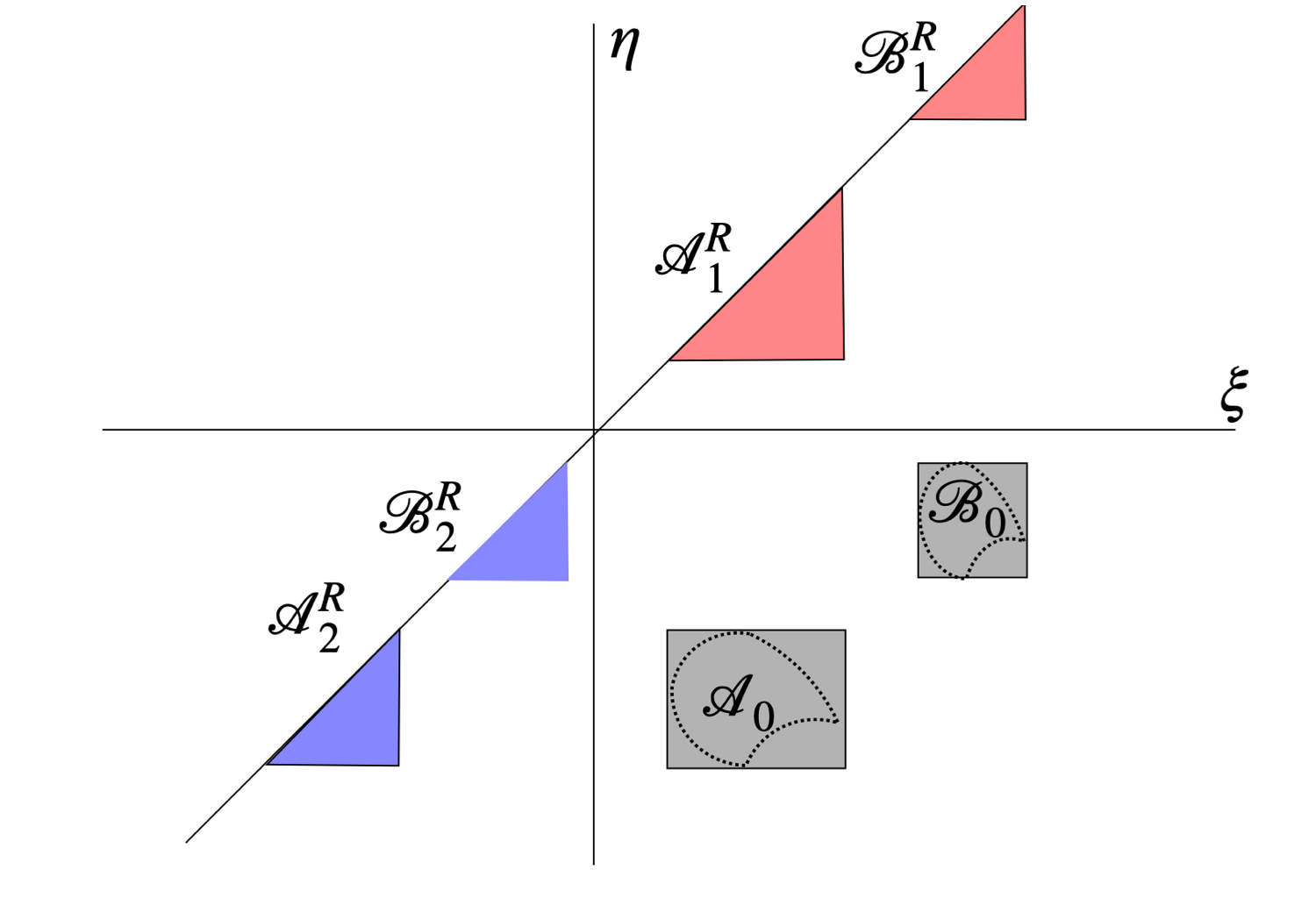}
    \caption{}
    \label{fig1:subfig2}
\end{subfigure}
\hfill
\begin{subfigure}{0.32\textwidth}
    \centering
    \includegraphics[width=\linewidth]{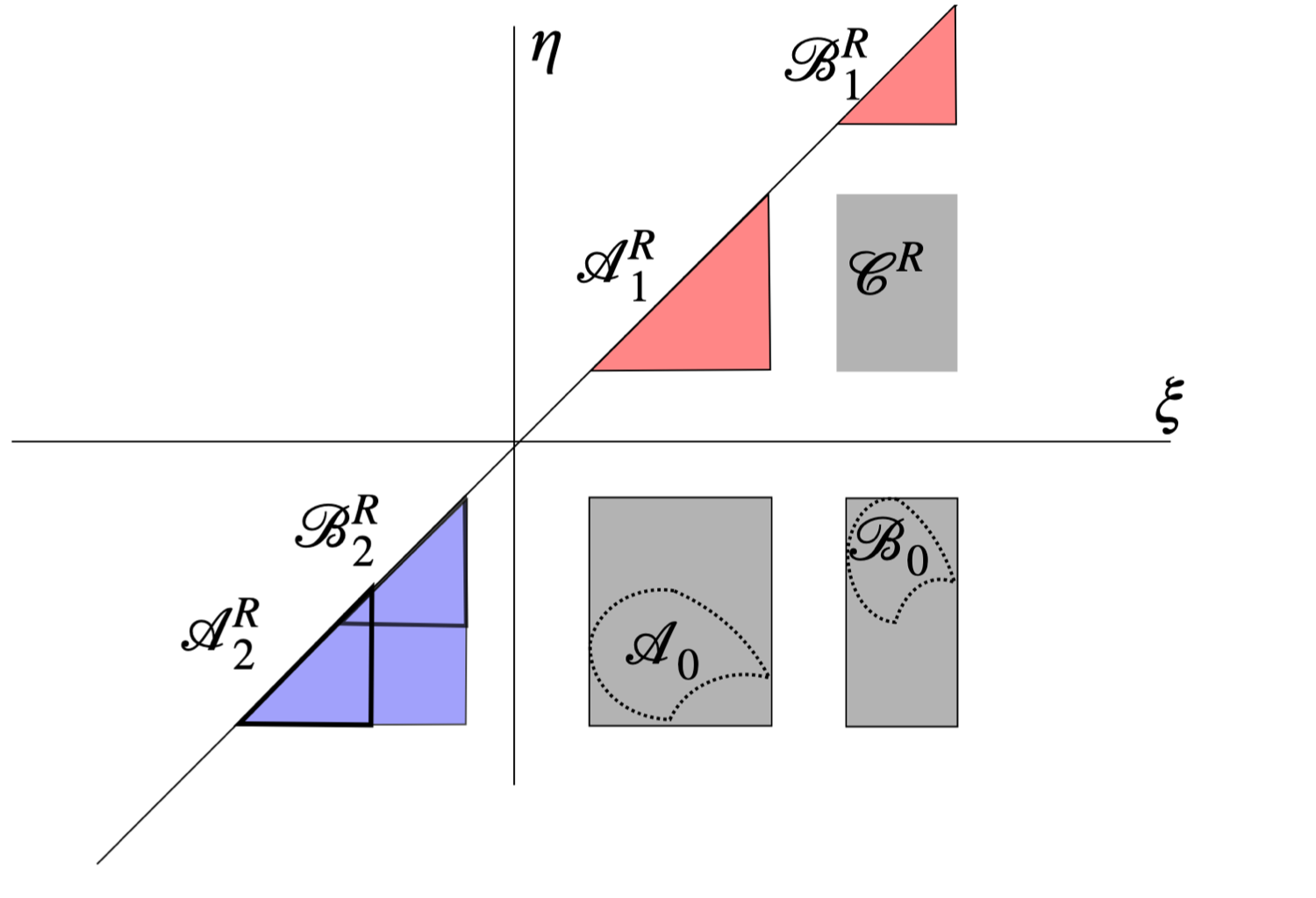}
    \caption{}
    \label{fig1:subfig3}
\end{subfigure}
\caption{(a) Schematic of a compact domain $\mathscr{D}_0$, its minimal enclosing rectangle $\mathscr{D}_0^{R}$, and the triangular regions $\mathscr{D}_1^{R}$ (red) and $\mathscr{D}_2^{R}$ (blue). The inset shows the allowed triadic interactions. (b) Schematic of two disconnected compact domains $\mathscr{A}_0$ and $\mathscr{B}_0$ with no intersecting projections along $\xi$ or $\eta$ and corresponding final supports (shaded in gray, blue and red). }
%\label{Fig:permittriads}
\end{figure*}

\section{Preservation of compact supports}

To address these questions, we first analyze the resonant manifold branches in Eq.~\eqref{Eq:rmsolutions}. Since $k_x>0$, all modes lie in $\mathbb{R}^2_{>} := \{ (\xi,\eta) \in \mathbb{R}^2 \mid \xi > \eta \}$. A resonant triad is given by $\{(\xi,\eta),(\xi,\xi_1),(\xi_1,\eta)\}$ (and symmetrically $\{(\xi,\eta),(\xi,\xi_2),(\xi_2,\eta)\}$), with
\begin{equation}\label{eq:triad_ineq}
    \xi \geq \eta,\quad \xi \geq \xi_1,\quad \xi_1 \geq \eta.
\end{equation}
These constraints restrict interactions to "inverted-L" configurations (see inset in Fig.~\ref{fig1:subfig1}). 

Moreover, as depicted in Fig.~\ref{fig1:subfig1} (i) if two modes of a resonant triad lie in the gray region, the third lies in one of the triangular regions, $\mathscr{D}_1^{R}$ or $\mathscr{D}_2^{R}$; (ii) if two modes lie within the same triangular region, the third remains in that region. Consequently, if the initial spectrum is supported on a compact domain $\mathscr{D}_0$, wave action spreads only into $\mathscr{D}_1^{R}$ and $\mathscr{D}_2^{R}$, while also filling the minimal enclosing rectangle $\mathscr{D}_0^{R}$. This observation can be made precise in the following theorem.

\begin{theorem}\label{th:compD0}
Let $n(\xi,\eta,t)$ be a solution of the wave kinetic equation on $\mathbb{R}^2_{>}$. Suppose
\[
\operatorname{supp} n(\cdot,0) = \mathscr{D}_0 \subset \mathbb{R}^2_{>},
\]
where $\mathscr{D}_0$ is a connected compact domain. Then for all $t > 0$,
%\textcolor{red}{Why need $t^*$? Then for any $t>0$ ...}
\[
\operatorname{supp} n(\cdot,t)
=
\mathscr{D}^{R}_0 \cup \mathscr{D}^{R}_1 \cup \mathscr{D}^{R}_2.
\]
\end{theorem}

The proof follows directly from the constraints in \eqref{eq:triad_ineq} and is provided in the Appendix.~\ref{methods:proof_fsupp}. The result extends to more general compactly supported initial data, including disconnected domains (see Appendix.~\ref{methods:proof_fsupp}). 
As an illustration, consider two initially disconnected domains, $\mathscr{A}_0$ and $\mathscr{B}_0$, with no overlap in their $\xi$- or $\eta$-coordinates. In this case, the resulting finite support is simply the union of the individual supports associated with each domain (i.e., four triangular regions and two rectangles; see Fig.~\ref{fig1:subfig2}).

Slightly more complex scenario arises when the domains have overlapping ranges in either $\xi$ or $\eta$. For instance, in Fig.~\ref{fig1:subfig2}, the domains $\mathscr{A}_0$ and $\mathscr{B}_0$ overlap in their $\eta$-coordinates. In such cases, Theorem~\ref{th:compD0} still applies; however, the final support is modified and corresponds to the shaded regions shown in Fig.~\ref{fig1:subfig3}.
This restricted mixing has important implications: the relevant RJ states must be supported on the same domain, i.e.,
\begin{equation}{\label{eq:RJsupport}}
\operatorname{supp} n_{\small{RJ}} = \mathscr{D}^{R}_0 \cup \mathscr{D}_1^{R} \cup \mathscr{D}_2^{R}
\end{equation}
in the cases of an initially single connected compact domain. For initial data consisting of two disconnected domains with no overlap in the $\xi$- or $\eta$-ranges (as in Fig.~\ref{fig1:subfig2}), the dynamics on the two corresponding supports, $\mathscr{A}^{R}_0 \cup \mathscr{A}_1^{R} \cup \mathscr{A}_2^{R}$ and $\mathscr{B}^{R}_0 \cup \mathscr{B}_1^{R} \cup \mathscr{B}_2^{R}$, remain independent. Consequently, the wave action thermalizes into distinct RJ spectra associated with $\mathscr{A}_0$ and $\mathscr{B}_0$, respectively.\\
Recall that usually for the classical fields, the RJ solutions are subject to the well-known ultraviolet catastrophe. For example, for the RJ solutions corresponding to states of energy equipartition among all Fourier modes, the physical-space energy density is infinite  when the Fourier space is unbounded. 
%, however, as $k\to\infty$, such equipartition at finite temperature necessarily implies an infinite total energy. 
This divergence severely limits the physical relevance of RJ equilibria.
In contrast, our results demonstrate that this difficulty is naturally circumvented in the present setting. Specifically, the system exhibits an effective \emph{self-truncation}, whereby the dynamics remain confined to a finite region of Fourier space determined by the support of the solution. As a consequence, the ultraviolet divergence is avoided, and the resulting equilibria remain physically well-defined.

%In the  case of multiple domains including non-intersecting enclosing rectangles, the independent dynamics in the   sets of supports (each consisting of a rectangle plus two triangles)  would generally evolve into  a distinct RJ spectrum in each such a set.

\begin{figure}[!h]
%\centering
\includegraphics[width=\linewidth]{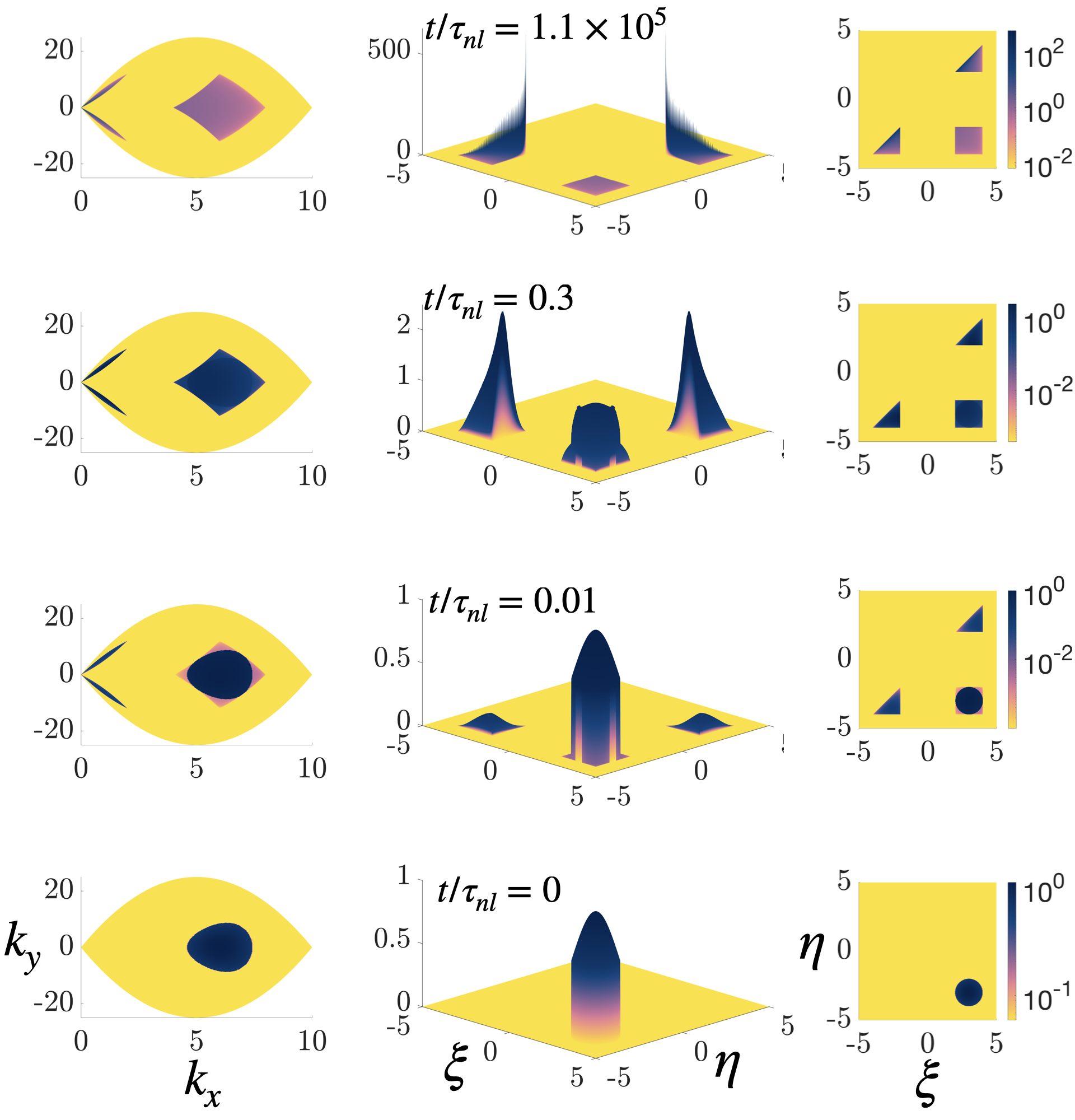}
\caption{Spectrum evolution in Run B. Left column: 2D color plot of $n(k_x,k_y)$. Center and right columns:
spectrum in the 
$(\xi,\eta)$ space -- the landscape and the 2D plots respectively.
Yellow area marks the computational domain; it is square $(-5,5)^2$ in the $(\xi,\eta)$ plane and a curved "eye" shape in $(k_x,k_y)$.}
\label{Fig:icevol}
\end{figure}
%\begin{figure*}[t]
%\centering
%\begin{subfigure}{0.48\textwidth}
 %   \centering
 %   \includegraphics[width=\linewidth]{figures/fig_2_maintext_invaspec/fig_2_maintext_invaspec.001.png}
 %   \caption{}
 %   \label{fig2:subfig1}
%\end{subfigure}
%\hfill
%\begin{subfigure}{2*\columnwidth}
%    \centering
%    \includegraphics[width=\linewidth]{figures/fig_2_maintext_invaspec/fig_2_maintext_invaspec.002.png}
%    \caption{}
%    \label{fig2:subfig2}
%\end{subfigure}
%\caption{ }
%\label{fig2}
%\end{figure*}

\section{Numerical Results}

We now present results from numerical simulations of the WKE~\eqref{Eq:wke_xieta}. 
 In this representation, the $\bm k$ and $\omega$ delta functions have already been integrated out, which allows to avoid the interpolation errors typical for the WKE simulations~\cite{krstulovic2025wavkins} and achieve a highly accurate conservation of the invariants Eq.~\eqref{Eq:inv} (with relative error $\sim 10^{-12}$).
 
 For initial condition we choose a Gaussian supported on $\mathscr{D}_0$:
\[
n(\xi,\eta,0)= \exp\!\left(-\frac{(\xi-\xi_0)^2}{\Delta_{\xi}^2}-\frac{(\eta-\eta_0)^2}{\Delta_{\eta}^2}\right), \quad (\xi,\eta)\in \mathscr{D}_0,
\]
where $(\xi_0,\eta_0)$ sets the peak location and $\Delta_{\xi},\Delta_{\eta}$ determine the widths.
We present results from two simulations which differ only in their initial support $\mathscr{D}_0$: for Run A $\mathscr{D}_0$ is a square and for Run  B $\mathscr{D}_0$ is a circle,
see Appendix.~\ref{methods:DNS_WKE} for details.
%bottom insets of Fig.~\ref{fig2:subfig1} and Fig.~\ref{fig2:subfig2}, respectively. 
Additionally, we define a kinetic (nonlinear) timescale $\tau_{nl}$ as the time at which the maximum of the initial spectrum decreases to three-quarters of its initial value.
In Fig.~\ref{Fig:icevol} we show the snapshots from the time evolution of the spectrum in Run B in both
$(k_x,k_y)$ and the $(\xi,\eta)$ spaces. Run A evolution is visually similar (see Appendix.~\ref{methods:add_sim}), but there are important quantitative differences at the thermalized stage, which will be discussed in the end of this section. We also present in Appendix.~\ref{methods:add_sim} a numerical run representing the case in Fig.~\ref{fig1:subfig3}, this run also shows onset of thermalization.

\subparagraph{Initial evolution and nonlocal transfer.} 
For both Run A and Run B, the early-time dynamics, at $t/\tau_{nl}\sim \mathcal{O}(10^{-2})$, exhibit a symmetric and nonlocal spreading of wave action into the regions $\mathscr{D}_1^{R}$ and $\mathscr{D}_2^{R}$ (see Fig.~\ref{Fig:icevol}, central and right columns). In $k$-space (Fig.~\ref{Fig:icevol}, left), these regions correspond to the beam-like structures observed in the low-$k$ regime, already highlighted in the Introduction in the context of direct numerical simulations of KP-I. We now identify these structures as a direct consequence of three-wave resonances, as captured by the WKE. The initial transfer is primarily mediated by triads involving modes in the immediate vicinity of the initial spectral peak, namely $\{(\xi_0,\eta_0), (\xi_0,\eta_1), (\xi_1,\eta_0)\}$.

Since $\Delta \ll |\xi_0|, |\eta_0|$, the initial evolution can be interpreted as a modulational instability of a single mode with ${\bf k} = (\xi_0 -\eta_0, \, 0)$ where the modulations of the length scale of $1/\Delta$ arise due to the
spectrum broadening with width $\sim \Delta$. This makes our initial evolution similar to the classical setup for the modulational instability of a periodic KP-1 wave \cite{kuznetsov1984stability,infeld2000nonlinear,modulational2017ablowitz}. However, the difference is that in our case the modulations are random, which follows from the randomness of phases of waves described by the WKE. Since the line $\xi=\eta$ corresponds to $\bf k =0$, we observe that the growth of the peaks close to this line is a result of an instability of the initial state characterized by transfer of its energy to the largest scales in a strongly nonlocal way (bypassing excitation of the intermediate scales). 

At intermediate times $t/\tau_{nl}\sim \mathcal{O}(1)$, peaks emerge along the line $\xi=\eta+c$, where $c$ is of the order of the grid spacing. (Note that $n(\xi,\eta,t)=0$ along $\xi=\eta$ for all times.) The evolution for Run A and Run B is visually similar, but the position of the peaks is somewhat different. This transient evolution marks the initial relaxation toward the thermalized state.

\subparagraph{Thermalization.}
At late times, $t/\tau_{nl} \sim \mathcal{O}(10^5)$, these peaks migrate toward the corners of the triangular regions $\mathscr{D}_1^{R}$ and $\mathscr{D}_2^{R}$, departing significantly from the initial resonant structure centered at $(\xi_0,\eta_0)$. Their amplitude increases to $\gtrsim \mathcal{O}(10^2)$ relative to the initial condition. In the $(k_x,k_y)$ representation, this corresponds to the emergence of a pronounced double-beam-like peak near $(k_x,k_y)\simeq(0,0)$ (see Fig.~\ref{Fig:icevol}, left).

In the long-time limit $t/\tau_{nl} \to \infty$, the wave action is expected to relax toward an RJ state given by Eq.~\eqref{Eq:rjsol}, where the functional form $F(x)$ is constrained by an infinite set of invariants determined by the initial condition. We demonstrate thermalization in both runs using three independent diagnostics.

For RJ equilibria the entropy production rate vanishes, $d\mathcal{S}/dt = 0$. Consistently, Fig.~\ref{Fig:entropy} shows that the entropy $\mathcal{S}$ from our simulations approaches a plateau at long times, supporting the onset of thermalization.

However, a more defining and stricter property of RJ solutions of the WKE is that the nonlinear collision integral vanishes pointwise: each square bracketed term in the integrand of Eq.~\eqref{Eq:wke_xieta} (or the nonlinear part of Eq.~\eqref{Eq:r12k}) goes to zero. To quantify this, we monitor
\begin{equation}
\mathscr{N}_{\infty}
=\sup_{\bm{\chi},\bm{\chi}_1,\bm{\chi}_2}
\big| n_{\bm \chi_1} n_{\bm \chi_2} - n_{\bm \chi} n_{\bm \chi_1} - n_{\bm \chi} n_{\bm \chi_2}\big|,
\end{equation}
where, $\bm {\chi}=(\xi,\eta)$. For an RJ state, $\mathscr{N}_{\infty}=0$. As shown in Fig.~\ref{Fig:entropy}, $\mathscr{N}_{\infty}$ starts at $\sim \mathcal{O}(10^3)$ and decreases with time reaching values $\lesssim 10^{-11}$ for Run A and $\lesssim 10^{-9}$ for Run B at $t/\tau_{nl} \sim \mathcal{O}(10^5)$, indicating convergence toward equilibrium.

Finally, for initial conditions symmetric about the line $\xi+\eta=0$ (equivalently $k_y=0$), the WKE preserves this symmetry for all $t>0$. Consequently, RJ solutions inherit this symmetry, implying that
\begin{equation}
F(\xi)-F(\eta)=F(-\eta)-F(-\xi).
\label{Eq:symfit}
\end{equation}
This relation enables reconstruction of $F(\cdot)$, and therefore $n_{RJ}$ of the form Eq.~\eqref{Eq:rjsol}, from the final simulation snapshot data along a single horizontal (or vertical) slice in $\mathscr{D}_0^{R}$ (see Appendix.~\ref{methods:FIT_RJ}). %Using the final simulation snapshot, we extract such a slice and construct a fit $\hat{n}(\xi,\eta)$ via \eqref{Eq:symfit}, which by construction satisfies 
%\eqref{Eq:rjsol}. 
The relative error between the resulting $n_{RJ}$ %defined in \eqref{Eq:rjsol} 
and the numerical $n(\xi,\eta)$ is at most $\sim 10^{-9}$, indicating convergence to the RJ form and providing further solid evidence of thermalization.
Moreover, sometimes  function $F(\cdot)$ obeying Eq.~\eqref{Eq:symfit}
can be accurately represented by a polynomial 
$F(x) \approx \sum_{j=1}^{2N+1} x^{2j+1}  $ leading to an accurate RJ fit. For example, $N=10$  such a fit in Run A leads to the maximum relative deviation between the numerical and the RJ spectrum being at most $\sim\mathcal{O}(10^{-2})$. Thus considering first few invariants with power law densities is enough to define the system, without worrying about the whole functional family of them. Unfortunately, not always the polynomial approximation provides a good fit, for instance it does not for our Run B, even though the functional fit works perfectly well.
Conditions for the polynomial fit to work are not yet completely clear.

\begin{figure}[!]
\centering
\includegraphics[width=\linewidth]{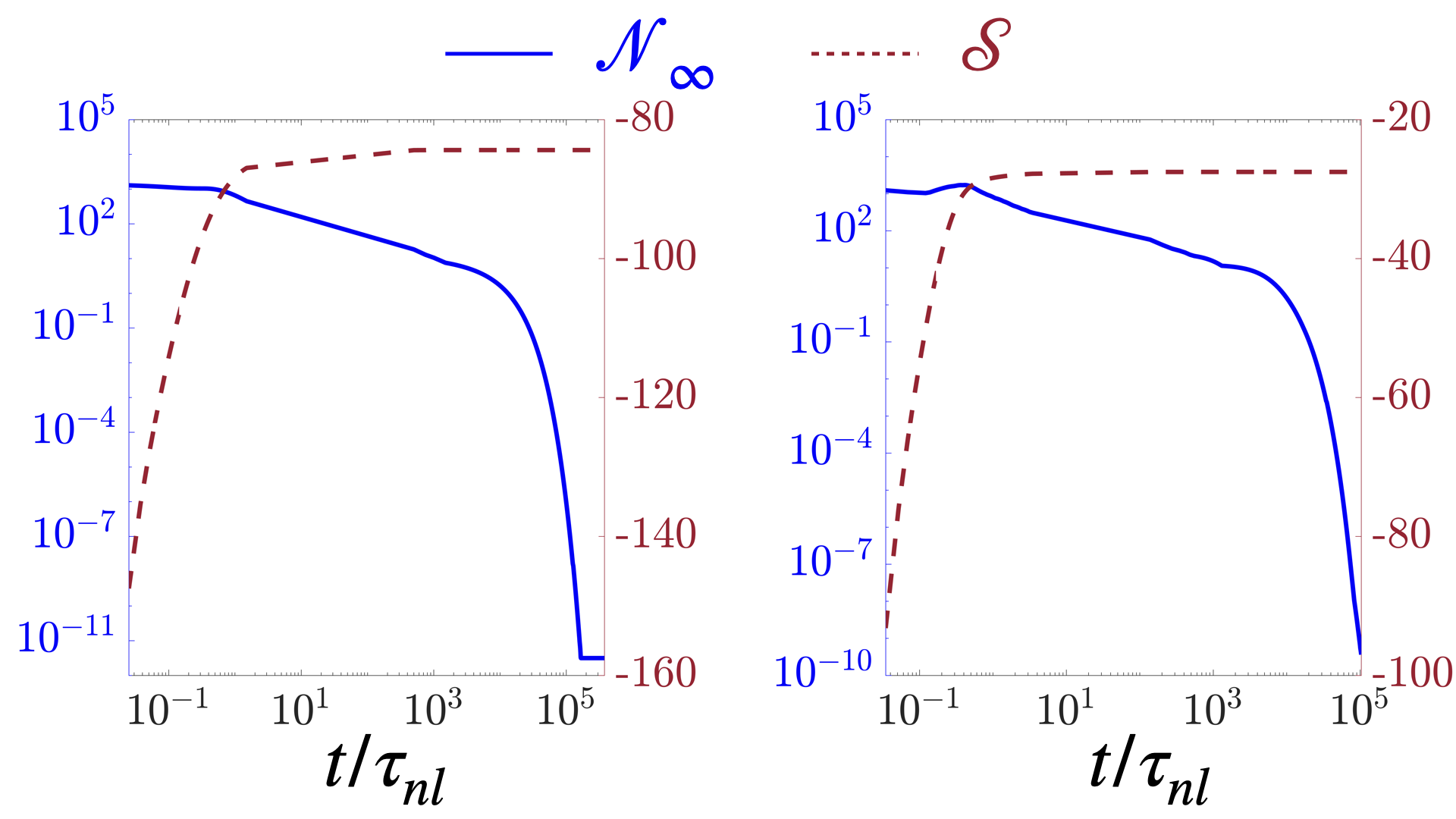}
\caption{Temporal evolution of bracket ${\cal N}_\infty$ (solid line) and entropy $\mathcal{S}$ (dashed line). Left: Run A; right: Run B. }
\label{Fig:entropy}
\end{figure}

\section{Summary and Discussions}

In this paper, we studied the Fourier  space evolution and thermalization described by the WKE derived from the integrable KP-I model. Our main results may be summarized as follows.

\begin{itemize}
    \item Despite its integrability, the KP-1 model allows transfers of energy and other quantities among the Fourier modes. Such transfers are described by the WKE of the wave turbulence theory.
    \item 
    The evolution described by the WKE leads to nonlocal generation of large-scale motions having a shape of two symmetric low-$k$  beams. 
    \item The WKE evolution preserves compactness of the $k$-space supports. 
    \item The long time evolution described by the WKE can lead to thermalization to a generalized RJ spectrum, whose shape is determined by the whole functional series of invariants Eq.~\eqref{Eq:inv}.
    In some cases, not always, only few basic invariants out of that family (the ones with power law densities) are sufficient for determining the thermal state with high accuracy.
    \item We proved thermalization to RJ states using two examples, but this does not rule out a possibility that there may exist such initial conditions that thermalization to the RJ is impossible. Potentially, the WKE evolution may lead to a spectrum condensation into a Dirac mass at an isolated $k$-mode, in a way it happens, e.g., for the NLS system \cite{Connaughton_CondensationClassicalNonlinear_2005}.
    This issue requires further study.
\end{itemize}
Now, we would like to discuss links with the KP-I evolution beyond the WKE description. Since the thermalized RJ states contain high low-$k$ peaks (infinite if we allow $k\to0$), the WKE description based on the weak nonlinearity assumption must inevitably break down near these peaks. The phases will no longer be random for low-$k$ modes, and coherent structures are likely to appear. Lump solutions represent an example of stable coherent structures in KP-I model, which arise, in particular, as a result of instability of the KP-I soliton \cite{Pelinovsky95}. In the physical space a single-lump solution is~\cite{Satsuma,ablowitz2011nonlinear}:
\begin{equation}{\label{Eq:lumpsol}}
    u(x',y')=4\frac{-(x'-2ay')^2+4b^2y'^2+\frac{1}{4b^2}}{\bigg[(x'-2ay')^2+4b^2y'^2+\frac{1}{4b^2}\bigg]^2}
\end{equation}
where, $a$ and $b$ are parameters and the shifted coordinates are defined as $x' = x - 12(a^2 + b^2)t$ and $y' = y - 12 a t$.
The waveaction spectrum corresponding to the state  with two spatially separated lumps with $a=3$ and $a=-3$
is shown in Fig.~\ref{Fig:lumps} in both the $(k_x,k_y)$ space (left panel) and the
$(\xi,\eta)$ space (right panel). We see a striking similarity in the spectra of the double-lump structure with the double-beam  $k$-space structure arising from the WKE evolution and corresponding double triangle structure in the $(\xi,\eta)$ plane.
This similarity leads to a conjecture that at the strongly nonlinear stage, when appearance of the strong peaks leads to failure of the WKE and the wave phases become correlated, localized lump-like coherent structures appear. The latter will retain the $k$-space features inherited from the previous stage governed by the WKE.    \\

The primary aim of this work has been to demonstrate that, contrary to conventional expectations, the KP-I integrable system admits a nontrivial wave turbulence description, characterized by genuine energy exchange between modes and the formation of long-wave structures that drive the system toward thermalization. Nevertheless, several natural questions arise, which merit further investigation.
First, what is the precise relationship between the Fourier-based description employed in the WTT and a representation more closely aligned with the IST? While no direct analogue of the squared eigenfunction basis—familiar from 1D integrable systems—is currently known for KP-I, it is natural to ask to what extent the observed energy transfer in the Fourier basis reflects a redistribution into a more appropriate nonlinear basis. Such a basis would include nonlinear normal modes, which, in the KP-I setting, correspond to lump-like rational solutions.
Second, our numerical results suggest that the dynamics are dominated by interactions between a localized packet of Fourier modes centered at $(k_{0x},k_{0y})$ and a long-wave component. This raises the question of whether a reduced long-wave/short-wave description can be constructed, potentially within an integrable framework, as is known for several such coupled systems. In this picture, the resonance condition in Eq.~\eqref{Eq:rmanifold} acquires a clear interpretation: the group velocity of the short-wave packet in the direction of the long wave matches the phase velocity of the long wave. Thus, the dynamics may be viewed as a short-wave packet propagating on a long-wave background. This interpretation is consistent with the Burgers-type structure of KP-I, where the nonlinear term $ u u_x$ leads to the formation of steep gradients that are subsequently regularized into dispersive shock waves with extended oscillatory tails.
Finally, it is natural to ask how these considerations extend to the KP-II equation. In that setting, the soliton dispersion relation (relating speed to width) coincides with the linear dispersion relation of KP-I. Moreover, the phase shifts arising in two-soliton interactions involve a resonance condition analogous to Eq.~\eqref{Eq:rmanifold}, but with nonlinear frequencies. In this context, resonance corresponds to the formation of a Mach stem—a structure that can be interpreted as an infinitely extended soliton generated by the interaction. This raises the intriguing possibility that KP-II dynamics may favor the generation of highly elongated solitons, capable of transporting significant energy along their crests. In physical contexts such as shallow water, such structures could, in principle, have dramatic consequences, potentially leading to extreme wave events, like massive tsunamis.

\begin{figure}[!]
\centering
\includegraphics[width=\linewidth]{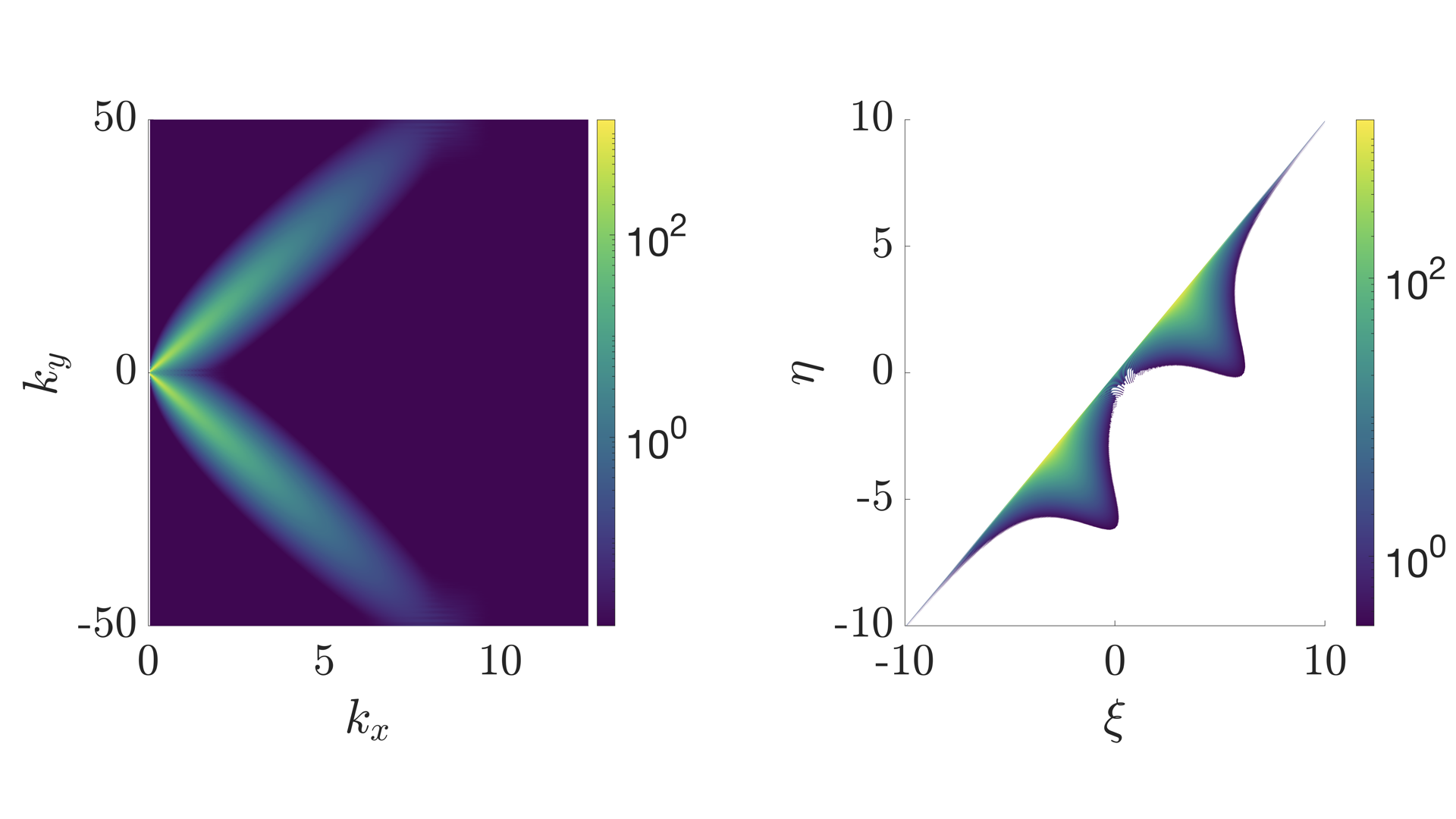}
\caption{Color log-scale plot of the wave-action spectrum corresponding to the superposition of two independent lump solutions of the KP-I equation. The right panel shows the spectrum in the $(k_x,k_y)$ plane, while the left panel displays it in the transformed $(\xi,\eta)$ coordinates.
 }
\label{Fig:lumps}
\end{figure}
\bibliographystyle{unsrt}
 \bibliography{references}
\appendix
\section{Numerical methods}
\subsection*{Direct numerical simulation of KP-1 equation} \label{methods:DNS_KP}
We solve Eq.~\eqref{Eq:KP} with $\sigma=1$ using a pseudo-spectral method with periodic boundary conditions on a $1024\times1024$ grid of collocation points. Time integration is performed using a standard fourth-order Runge–Kutta scheme, supplemented with a $1/3$ dealiasing rule. The time step is set to $\Delta t = 2\times10^{-4}$, which is smaller than the shortest linear timescale $\tau_L = 2\pi/\omega(\mathbf{k})$.
The initial condition is prescribed in Fourier space as
\[
\hat{u}(\mathbf{k},0)=10^{-1}\cdot \exp^{i\phi_{\bm{k}}}, 
\]
for $\bm{k}$ lying within a circular region defined by

\[
(k_x - 2)^2 + k_y^2 \leq 0.5^2
\]
and zero otherwise. Here, $\phi_{\mathbf{k}}$ are independent random phases uniformly distributed in $[0,2\pi)$. In other words, the initial energy is concentrated on a disk (circle) in Fourier space centered at $(2,0)$. 
\subsection*{Direct numerical simulations of the WKE}\label{methods:DNS_WKE}
Eq.~(\ref{Eq:wke_xieta}) is solved on a uniform grid over the domain $[\xi_{\min},\xi_{\max}]\times[\eta_{\min},\eta_{\max}]$, discretized with $512\times512$ modes. Since $\xi>\eta$ (equivalently $k_x>0$), modes satisfying $\eta>\xi$ are excluded from the computational domain. Time integration is carried out using a second-order Runge–Kutta scheme with a time step $\Delta t = 2\times10^{-4}$. All integrals are evaluated using the trapezoidal rule.

The simple structure of the resonant manifold eliminates the need for interpolation schemes—commonly required in other WKEs—and enables conservation of invariants to high accuracy. In all simulations, even at long times $t/\tau_{nl}\sim \mathcal{O}(10^{5})$, the relative error in conserved quantities (with respect to their initial values) remains of order $\sim 10^{-12}$.

As discussed in the main text, we consider Gaussian initial conditions supported on compact domains. For Run A, the support is a square region
\[
\mathscr{D}_0 = \{(\xi,\eta): 1.5 \leq \xi \leq 4.5,,-4.5 \leq \eta \leq -1.5\},
\]
with a Gaussian profile centered at $(\xi_0,\eta_0)=(3,-3)$ and widths $\Delta_\xi=\Delta_\eta=0.8$.

For Run B, the support is chosen to be a disk,
\[
\mathscr{D}_0 = \{(\xi,\eta):(\xi-3)^2 + (\eta+3)^2 \leq 1\},
\]
centered at $(3,-3)$. The initial Gaussian is again centered at $(3,-3)$ with $\Delta_\xi=\Delta_\eta=1.0$.

In both cases, the widths satisfy $\Delta < |\xi_0|, |\eta_0|$, ensuring that the initial distribution is well localized in spectral space.

\subsection*{Fitting for Rayleigh-Jeans}\label{methods:FIT_RJ}
The intergals
\[
%I_{2n+1}=\iint_{\mathscr{D}_0} 2(\xi-\eta)(\xi^{2n+1}-\eta^{2n+1}) n(\xi,\eta) d\xi d\eta=0, ~~n\in \mathbb{N},
I=\iint 2(\xi-\eta)(F(\xi)-F(\eta) )\,n(\xi,\eta,t) \,d\xi
\]
must be conserved for all times for arbitrary $F(\cdot)$. Out of the (functionally) infinite possibilities, a single choice of $F(\cdot)$ will correspond to an invariant which be equi-partitioned in the RJ spectrum Eq.~\eqref{Eq:rjsol} to which the system is relaxing. 

Our task here is devise a fitting procedure to find such $F(\cdot)$ in the special case when the initial data, and therefore the spectra at any $t>0$, are mirror symmetric with respect to $k_y$, i.e. $n(k_x,-k_y,t) = n(k_x,k_y,t)$. This means that $n(\xi,\eta)=n(-\eta, -\xi)$.

Respectively, the RJ solution Eq.~\eqref{Eq:rjsol} must also preserve this symmetry and, therefore,
\[
F(\xi)-F(\eta)=F(-\eta)-F(-\xi)
\]
for all $(\xi,\eta)\in  \mathscr{D}^{R}_0 \cup \mathscr{D}_1^{R} \cup \mathscr{D}_2^{R}$. 
For such symmetric data, $\mathscr{D}_0$ becomes a square,  and triangular supports are equal in size:
\[
\mathscr{D}^{R}_0
=
\left\{
(\xi,\eta) \in \mathbb{R}^2_{>}
\;\middle|\;
a \leq \xi \leq b,\; -b \leq \eta \leq -a,\; 
\right\},
\]
\[
\mathscr{D}^{R}_1
=
\left\{
(\xi,\eta) \in \mathbb{R}^2_{>}
\;\middle|\;
a \leq \xi \leq b,\; a \leq \eta \leq \xi,\; 
\right\},
\]
\[
\mathscr{D}^{R}_2
=
\left\{
(\xi,\eta) \in \mathbb{R}^2_{>}
\;\middle|\; -b\leq\eta \leq \xi , -b \leq\xi \leq -a
\right\}.
\]
%Since $\mathscr{D}^{R}_0$ is the enclosing rectangle, we know $\mathscr{D}^{R}_0=[a,b]\times[-b,-a]$ with the sides determined by the initial support $\mathscr{D}_0$, in fact $\mathscr{D}^{R}_0$ is a rectangle. 
For any $(\xi,\eta_0)\in\mathscr{D}^{R}_{0}$ we then have
\[
F(\xi)-F(\eta_0)=F(-\eta_0)-F(-\xi), ~~~~-b<\eta_0<-a~ \; \& \; ~ a<\xi<b,
\]
Then inside $\mathscr{D}^{R}_0$, for any $\eta$, we have
\[
F(-\eta)=-F(\eta)+F(\eta_0)+F(-\eta_0), ~~~~-b<\eta<-a
\]
Using  this relation, we can write for $(\xi,\eta)\in \mathscr{D}^{R}_0$:
\begin{equation}
  \label{Eq:RJfit}  
F(\xi)-F(\eta)=\tilde{F}(\xi)+\tilde{F}(-\eta)-\tilde{F}(-\eta_0),
\end{equation}
where we have introduced $\tilde{F}(\cdot)=F(\cdot)-F(\eta_0)$. Now from the simulations, for example taking the last snapshots, we have $n(\xi,\eta)$, and we will be looking for $F(\cdot)$ which would provide a fit $1/n(\xi,\eta) \myeq F(\xi)-F(\eta)$. We then choose an arbitrary $-b\leq\eta_0\leq-a$
and, using Eq.~\eqref{Eq:RJfit}, find
%\begin{equation}
%  \label{Eq:RJfit1}  
$\tilde{F}(\cdot)\myeq %\tilde{F}(\eta_0)-\tilde{F}(-\eta_0) 
1/n(\cdot,\eta_0)$ and
%\end{equation}
\begin{equation}
  \label{Eq:RJfit1} 
  F(\xi)-F(\eta) \myeq
  1/n(\xi,\eta_0) + 1/n(-\eta,\eta_0) - 1/n(-\eta_0,\eta_0).
\end{equation}
Thus, we have succeeded to fit the RJ spectrum $n_{RJ} (\xi,\eta) =1/(F(\xi)-F(\eta))$ in $\mathscr{D}^{R}_{0}$ using the data for the spectrum on a single slice $\eta=\eta_0, \, a<\xi<b$.

Now, for any $(\xi,\eta)\in\mathscr{D}^{R}_{1}$ we then have
\[
F(\xi)-F(\eta)=
\tilde F(\xi)- \tilde F(\eta)
, ~~~~a<\eta<\xi~ \; \& \; ~ a<\xi<b,
\]
where $\tilde F(\cdot)$ is defined with the same (negative)  $\eta_0$ as above. 
Then inside $\mathscr{D}^{R}_1$ we have
\begin{equation}
  \label{Eq:RJfit2} 
  F(\xi)-F(\eta) \myeq
  1/n(\xi,\eta_0) + 1/n(\eta,\eta_0).
\end{equation}
Again, to find the RHS we only need the spectrum on the same slice $\eta=\eta_0, \, a<\xi<b$.

Finally, the fit for  $(\xi,\eta)\in\mathscr{D}^{R}_{2}$
is obtained from the above fit for   $(\xi,\eta)\in\mathscr{D}^{R}_{1}$
via using the symmetry
$n(\xi,\eta)=n(-\eta, -\xi)$.
%\appendix
\section{Proofs for persistence of finite supports for Rayleigh-Jeans solutions}\label{methods:proof_fsupp}
In this section, we provide proofs of the theorems presented in the main text, along with several additional results corresponding to different choices of initial domains.
The overall strategy is as follows: we first analyze the primary resonances generated by the initial data, and then demonstrate that the secondary resonances arising from these interactions do not propagate beyond specified finite support.
It is useful to recall the resonant manifold solutions introduced in the main text:
\begin{align}
&\xi_1=\eta_2,\quad \xi_2=\xi,\quad \eta_1=\eta, \
&\xi_2=\eta_1,\quad \xi_1=\xi,\quad \eta_2=\eta.
\end{align}
Throughout this section, we focus on resonant triads of the first type, $\{(\xi,\eta),(\xi_1,\eta),(\xi,\xi_1)\}$. However, identical conclusions hold for the second class of solutions, as they are related by symmetry under the interchange of indices $1 \leftrightarrow 2$.
%\begin{t}[Confinement of spectral spreading to explicit resonant domains]
We first deal with the case of an initial support of rectangle, in the next sections we generalize to arbitary domains. We begin with a few definitions:
\begin{definition}{\label{def:supports}}
 Let
\[
\mathbb{R}^2_{>} := \{ (\xi,\eta) \in \mathbb{R}^2 \mid \xi > \eta \},
\]
and $\mathscr{D}^{R}_0\subset \mathbb{R}^2_{>}$ be a rectangular region in $\mathbb{R}^2_{>}$, that is,
\[
\mathscr{D}^{R}_0
=
\left\{
(\xi,\eta) \in \mathbb{R}^2_{>}
\;\middle|\;
a \leq \xi \leq b,\; c \leq \eta \leq d,\; a>d
\right\},
\]
Furthermore we also define the regions $\mathscr{D}^{R}_1$ and $\mathscr{D}^{R}_2$ as:
\[
\mathscr{D}^{R}_1
=
\left\{
(\xi,\eta) \in \mathbb{R}^2_{>}
\;\middle|\; a\leq\xi \leq b, c\leq\eta\leq\xi
\right\}.
\]
and
\[
\mathscr{D}^{R}_2
=
\left\{
(\xi,\eta) \in \mathbb{R}^2_{>}
\;\middle|\; c\leq\eta \leq \xi \leq d
\right\}.
\]
\end{definition}
Note that sets $\mathscr{D}^{R}_0$, $\mathscr{D}^{R}_1$ and $\mathscr{D}^{R}_2$ are the precise definitions of the shaded regions shown in Fig~\ref{Fig:finitesupport}.\\

Primary resonances stem from the initial support, they only possible resonances are given in the first of the inset in Fig~\ref{Fig:finitesupport}.\\ More formally we have the following lemma.

\begin{lemma}\label{lem:D0}
If two modes of a resonant triad  
\[\mathscr{R}\equiv\{(\xi,\eta),(\xi_1,\eta),(\xi,\xi_1)\}\subset\mathbb{R}^2_{>}\]
lie in $\mathscr{D}^{R}_{0}$, then the third mode necessarily lies in $\mathscr{D}^{R}_1 \cup \mathscr{D}^{R}_2$.\\
Moreover, the set of all such resulting modes spans the entire region $\mathscr{D}^{R}_1 \cup \mathscr{D}^{R}_2$. Specifically, if
\[\{(\xi_1,\eta_1),(\xi_2,\eta_2),(\xi_3,\eta_3)\}\]
denotes the three distinct modes of a resonant triad, then the collection of modes $(\xi_3,\eta_3)$ generated by all pairs $(\xi_1,\eta_1), (\xi_2,\eta_2) \in \mathscr{D}^{R}_0$ is given by
\[
\mathscr{F}\equiv\bigcup_{(\xi_1,\eta_1),(\xi_2,\eta_2)\in\mathscr{D}^{R}_0}(\xi_3,\eta_3)=\mathscr{D}^{R}_1\cup\mathscr{D}^{R}_2.
\]

\end{lemma}
\begin{proof}

For the first part of the proof we note that there are three possible configurations, the first is when  $(\xi,\xi_1)\in\mathscr{D}^{R}_0$ and $(\xi_1,\eta)\in\mathscr{D}^{R}_0$, which implies $a \le \xi_1 \le b$ and simultaneously $c \le \xi_1 \le d$. This is impossible as we have imposed $a>d$. Therefore, this configuration is geometrically ruled out. The second configuration is when $(\xi,\eta)\in\mathscr{D}^{R}_0$ and $(\xi_1,\eta)\in \mathscr{D}^{R}_0$. We have $a \leq \xi \leq b$ and $a \leq \xi_1 \leq b$. Moreover, since all modes are in $\mathbb{R}^2_{>}$ we have $\xi \geq \xi_1$ and therefore have 
\[
a \leq \xi < b, \quad a \leq \xi_1 \leq \xi,
\]
implying  $(\xi,\xi_1)\in\mathscr{D}^{R}_1$.
Similarly, for the final configuration, $(\xi,\eta)\in\mathscr{D}^{R}_0$ and $(\xi,\xi_1)\in\mathscr{D}^{R}_0$ and therefore
\[
c \leq \eta \leq d, \qquad c\leq\eta \leq \xi_1 \leq d
\]
implying  $(\xi_1,\eta)\in\mathscr{D}^{R}_2$.\\
This proves the first part of the lemma.\\

For the second part of the lemma, we need to show that the range of resonant modes generated by triads with two modes in $\mathscr{D}^{R}_{0}$ covers the entire region $\mathscr{D}^{R}_1 \cup \mathscr{D}^{R}_2$. To establish this, it suffices to first consider a resonant triad in which two modes lie on an arbitrary horizontal line segment within the square, with endpoints on its boundaries: consider resonant triads of the kind $\{(\xi,\eta^{\ast}),(\xi_1,\eta^{\ast}),(\xi,\xi_1)\}$ for an arbitrary fixed $\eta^{\ast}$. We already know if 
such that  $(\xi,\eta^{\ast})\in\mathscr{D}^{R}_0$ and $(\xi_1,\eta^{\ast})\in\mathscr{D}^{R}_0$, we know that $(\xi,\xi_1)\in\mathscr{D}^{R}_1$. Moreover, for all possible $(\xi,\xi_1)$ form the set
\[
\mathscr{F}_1\equiv\bigcup_{\substack{\xi\in[a,b]\\{\xi_1\in[a,\xi]}}} (\xi,\xi_1)=\mathscr{D}^{R}_1
\]\\
the couple $(\xi,\eta^{\ast})$,$(\xi_1,\eta^{\ast})$  is in $\mathscr{D}^{R}_0$.

Note that since $\eta^{\ast}$ is arbitrary this is valid for every horizontal line segment with endpoints on its boundaries of the square. Therefore, each of such line segments generates the entire region $\mathscr{D}^{R}_1$ by triad interactions. An extension of this for arbitrary line segments is dealt with later. \\
Similarly, if we consider triads consider a resonant triad in which two modes lie on an arbitrary vertical line segment$\{(\xi^{\ast},\eta),(\xi^{\ast},\xi_1),(\xi_1,\eta)\}$
\[
\mathscr{F}_2\equiv\bigcup_{\substack{\xi_1\in[\eta,d]\\{\eta\in[c,d]}}} (\xi_1,\eta)=\mathscr{D}^{R}_2.
\]
For any triad \[\{(\xi_1,\eta_1),(\xi_2,\eta_2),(\xi_3,\eta_3)\}\]
denoting the three distinct modes of a resonant triad, the collection of modes $(\xi_3,\eta_3)$ generated by all pairs $(\xi_1,\eta_1), (\xi_2,\eta_2) \in \mathscr{D}^{R}_0$ is $\mathscr{F}$ and since 
$\mathscr{F}_1\subset\mathscr{F}$ and $\mathscr{F}_2\subset\mathscr{F}$ we have 
\[
\mathscr{F}=\mathscr{D}^{R}_1\cup\mathscr{D}^{R}_2.
\]

\end{proof}
We now turn to the secondary resonances that arise when the interacting modes are allowed to lie in the extended set $\mathscr{D}^{R}_{0} \cup \mathscr{D}^{R}_{1} \cup \mathscr{D}^{R}_{2}$. 
As a consequence, we obtain the following lemma.
\begin{lemma}\label{lem:D012}
If two modes of the resonant triad  $\{(\xi,\eta),(\xi,\xi_1),(\xi_1,\eta)\}$ are in $\mathscr{D}^{R}_{0}\cup\mathscr{D}^{R}_{1}\cup\mathscr{D}^{R}_{2}$, then the third mode lies in $\mathscr{D}^{R}_0\cup\mathscr{D}^{R}_1\cup\mathscr{D}^{R}_2$
\end{lemma}
\begin{proof}
 We first list all possible configurations:
\begin{enumerate}[label=(\roman*)]
   % \item If two modes of the triad belong to $\mathscr{D}^{R}_0$ then the third mode must belong to $\mathscr{D}^{R}_1\cup\mathscr{D}^{R}_2$;
    \item If two modes of the triad belong to $\mathscr{D}^{R}_1$ then the third mode must belong to $\mathscr{D}^{R}_1$;
    \item If two modes of the triad belong to $\mathscr{D}^{R}_2$ then the third mode must belong to $\mathscr{D}^{R}_2$;
    \item If one mode belongs to $\mathscr{D}^{R}_0$ and the other to $\mathscr{D}^{R}_2$, then the third mode must belong to $\mathscr{D}^{R}_0$;
    \item If one mode belongs to $\mathscr{D}^{R}_0$ and the other to $\mathscr{D}^{R}_1$, then the third mode must belong to $\mathscr{D}^{R}_0$;
    
\end{enumerate}
The proofs of all the above results are elementary and follow from the defining inequalities of the regions $\mathscr{D}^{R}_1$, $\mathscr{D}^{R}_2$, and $\mathscr{D}^{R}_0$. We give just the proofs for (i) and (iii). 
For Case~(i), again as before, there are three possible configurations to select the two resonant modes. Consider the configuration in which $(\xi,\xi_1)\in\mathscr{D}^{R}_1$ and $(\xi_1,\eta)\in\mathscr{D}^{R}_1$, then the defining inequalities for $\mathscr{D}^{R}_1$ imply 
\[
a\leq\xi\leq b; \quad a\leq \xi_1\leq \xi; \quad a\leq \eta\leq \xi_1
\]
It follows immediately that
\[
a \leq \xi \leq b, \quad \quad a \leq \eta \leq \xi,
\]
and hence $(\xi,\eta) \in \mathscr{D}^{R}_1$. For configuration 2 we have $(\xi,\eta)$ and $(\xi_1,\eta)$ belong to $\mathscr{D}^{R}_1$, then 
\[
a\leq \xi\leq b; \qquad a\leq \xi_1\leq b
\]
and since $(\xi,\xi_1)\in \mathbb{R}^2_{>}$, we have
\[
a\leq \xi\leq b, \quad a\leq \xi_1\leq \xi,
\]
 therefore $(\xi,\xi_1)\in \mathscr{D}^{R}_1$. Similarly, if $(\xi,\eta)$ and $(\xi,\xi_1)$ belong to $\mathscr{D}^{R}_1$, then $(\xi_1,\eta)\in\mathscr{D}^{R}_1$. An analogous argument applies to Case~(ii). For case (iii), once again the only possible configuration is $(\xi,\eta)\in\mathscr{D}^{R}_0$ and $(\xi_1,\eta)\in\mathscr{D}^{R}_2$, which gives us
 \[
    c\leq \eta\leq \xi_1\leq d,\quad a\leq \xi\leq b
 \]
 and therefore, $(\xi,\xi_1)\in \mathscr{D}^{R}_0$. Similar arguments also prove case~(iv).
 \end{proof}
  \begin{figure}[!]
\centering
\includegraphics[width=\linewidth]{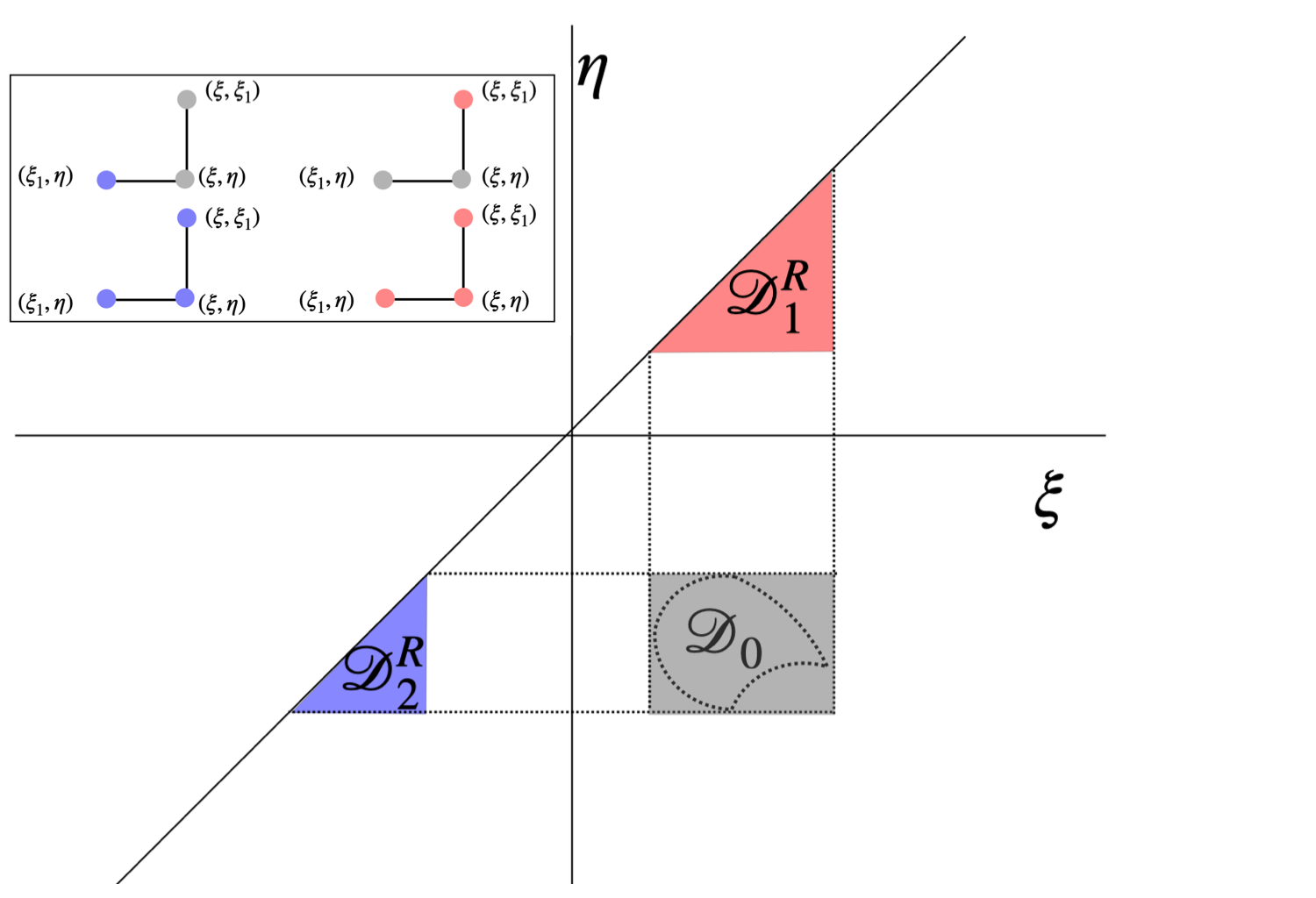}
\caption{ Schematic of a rectangular domain $\mathscr{D}^{R}_0$ (shaded gray), and corresponding triangles $\mathscr{D}^{R}_1$ and $\mathscr{D}^{R}_2$. The inset gives the possible primary and secondary resonances.
}
\label{Fig:finitesupport}
\end{figure}
%%%%%%%%%%%%%%%%

\newpage
 Now we can prove the following theorem

\begin{theorem}\label{th:fsuppsquare}
Let $n(\xi,\eta,t)$ be a solution of a wave kinetic equation
in $\mathbb{R}^2_{>}$. If the initial wave action has bounded support
\[
\operatorname{supp} n(\cdot,0) = \mathscr{D}^{R}_0 \subset \mathbb{R}^2_{>},
\]
 Then for all $t \in \mathbb{R}^+$,
\[
\operatorname{supp} n(\cdot,t)
=
\mathscr{D}^{R}_0\cup \mathscr{D}^{R}_1 \cup \mathscr{D}^{R}_2 .
\]

\end{theorem}
\begin{proof}
 The proof follows directly from lemmas ~\ref{lem:D0} and~\ref{lem:D012}. Suppose that at the initial time $t=0$, 
\[
n(\xi,\eta) = 0 \quad  (\xi,\eta) \in \mathbb{R}^2_{>}.
\]  
The only way this mode can get excited is through interactions with two other resonant modes for which the wave action is non-zero. At $t=0$, we assume that the support of $n$ is 
\[
\operatorname{supp} n(\cdot,0) = \mathscr{D}^{R}_0 \subset \mathbb{R}^2_{>},
\]  
so any two non-zero modes initially lie in $\mathscr{D}^{R}_0$. Therefore, by Lemma~\ref{lem:D0}, for small times $t$, the support evolves as
\[
\operatorname{supp} n(\cdot,t) \subset  \mathscr{D}^{R}_1 \cup \mathscr{D}^{R}_2.
\]  
Furthermore, Lemma~\ref{lem:D012} ensures that secondary interactions among modes in $\mathscr{D}^{R}_0 \cup \mathscr{D}^{R}_1 \cup \mathscr{D}^{R}_2$ do not generate new regions. Consequently, this containment holds for all times:
\[
\operatorname{supp} n(\cdot,t)= \mathscr{D}^{R}_0 \cup \mathscr{D}^{R}_1 \cup \mathscr{D}^{R}_2.
\]  

\end{proof}
The following remarks are in order: for regions defined in def.~\ref{def:supports}, the regions are pairwise disjoint, that is, 
\[
\mathscr{D}^{R}_0 \cap \mathscr{D}^{R}_1 = \mathscr{D}^{R}_0 \cap \mathscr{D}^{R}_2 = \mathscr{D}^{R}_1 \cap \mathscr{D}^{R}_2 = \emptyset.
\]
This behavior arises from the assumption (a>d), under which the corresponding regions remain disjoint. When this condition is relaxed, the regions overlap and the gaps in the wave-action spectrum disappear. In this case, additional resonant triads become accessible: if two modes of a triad lie within $\mathscr{D}_0^{R}$, then the third mode must also lie within $\mathscr{D}_0^{R}$. Nevertheless, the theorem stated above remains valid, with a redefined support $\mathscr{D}^{\ast}$:
\[
\operatorname{supp} n(\cdot,t) \subset  \mathscr{D}^{\ast}\equiv
\left\{
(\xi,\eta) \in \mathbb{R}^2_{>}
\;\middle|\; d\leq\xi\leq b, d\leq\eta\leq\xi
\right\}.
\]

Now in general $\mathscr{D}_0$ does not have to be a square but any compact subset of $\mathbb{R}^2_{>}$ with a boundary $\partial\mathscr{D}_0\subset\mathbb{R}^2_{>}$. In such cases, the proof has more steps, but straightforward. We prove it first for a connected compact domain, $\mathscr{D}_0$.\\ 
\begin{definition}\label{def:L_seg}
A vertical line segment (parallel to $\xi=0$) is defined by
\[
V(\xi_0; b_1,b_2)
:= \{(\xi_0,\eta)\in\mathbb{R}^2_{>} \mid b_1 \leq \eta \leq b_2\}.
\]

A horizontal line segment (parallel to $\eta=0$) is defined by
\[
H(\eta_0; b_1,b_2) 
:= \{(\xi,\eta_0)\in\mathbb{R}^2_{>} \mid b_1 \leq \xi \leq b_2\}.
\]
\end{definition}

\begin{lemma}{\label{lem:dense_L}}
Let $\mathscr{D}_0\subset\mathbb{R}^2_>$ be a compact connected domain in $\mathbb{R}^2_{>}$, with $\partial\mathscr{D}_0$ defining the boundary points. If two modes of the resonant triad  $\mathscr{R}\equiv\{(\xi,\eta),(\xi,\xi_1),(\xi_1,\eta)\}\subset\mathbb{R}^2_{>}$ lie in $V(\xi_0; b_1,b_2)\subset\mathscr{D}_{0}$, then the third mode lies in
\[
\mathscr{P}_1=
\left\{
(\xi,\eta) \in \mathbb{R}^2_{>}
\;\middle|\; b_1\leq \eta < \xi \leq b_2<\xi_0
\right\}.
\]

\end{lemma}
\begin{proof}
    Given that two resonant modes lie in $V(\xi_0; b_1,b_2)$, which is a vertical line segment, the only resonances allowed are the kind $\{(\xi_0,\eta),(\xi_0,\xi_1),(\xi_1,\eta)\}$, which directly implies 
    \[
    b_1\leq\eta\leq\xi_1\leq b_2<a,
    \]
    therefore, $(\xi_1,\eta)\in\mathscr{P}_1$.
\end{proof}
The minimum enclosing rectangle corresponding to $\mathscr{D}_0$ is defined as follows
\begin{definition}
    Let 
    \[
    \xi_{\min}=\inf \{\xi:~(\xi,\eta)\in\mathscr{D}_0\},\quad \xi_{\max}=\sup \{\xi:~(\xi,\eta)\in\mathscr{D}_0\},\quad
    \]
    \[
    \eta_{\min}=\inf \{\eta:~(\xi,\eta)\in\mathscr{D}_0\},\quad \eta_{\max}=\sup \{\eta:~(\xi,\eta)\in\mathscr{D}_0\},\quad
    \]
    then, the smallest rectangle enclosing $\mathscr{D}_0$ is 
    \[
    \mathscr{D}^{R}_0=\{(\xi,\eta)\in\mathbb{R}^2_{>}\mid \xi_{\min} \leq \xi \leq \xi_{\max},\; \eta_{\min} \leq \eta \leq \eta_{\max}\}.
    \]
    Like before, we also define $\mathscr{D}^{R}_1$ and $\mathscr{D}^{R}_2$ as follows:
    \[
    \mathscr{D}^{R}_1 =
\left\{
(\xi,\eta) \in \mathbb{R}^2_{>}
\;\middle|\;
\xi_{\min} \leq \xi \leq \xi_{\max},\; \xi_{\min} \leq \eta \leq \xi
\right\},
\]
and
\[
\mathscr{D}^{R}_2
=
\left\{
(\xi,\eta) \in \mathbb{R}^2_{>}
\;\middle|\; \eta_{\min}\leq\eta \leq \xi \leq\eta_{\max}
\right\}.
\]
\end{definition}

\begin{theorem}{\label{th:compD0}}

Let $n(\xi,\eta,t)$ be a solution of a wave kinetic equation
in $\mathbb{R}^2_{>}$. Let
\[
\operatorname{supp} n(\cdot,0) = \mathscr{D}_0 \subset \mathbb{R}^2_{>},
\]
where, $\mathscr{D}_0$ is compact connected domain with no cut points.
Then for $t > 0$, 
\[
\operatorname{supp} n(\cdot,t)
=
{\mathscr{D}^{R}_0} \cup \mathscr{D}^{R}_1 \cup \mathscr{D}^{R}_2 .
\]
\end{theorem}
\begin{proof}

We prove the theorem in two steps. In the first step, we show how wave action spreads to the regions
$\mathbb{R}^{2}_{>} \setminus \mathscr{D}^{R}_{0}$ and $\mathscr{D}^{R}_{0} \setminus \mathscr{D}_{0}$.
In the second step, we use these results to iteratively establish the full statement of the theorem. Schematically the proof follows from Fig.~\ref{Fig:proof_sketch}: , $V(a: b_1,b_2)$ and $V(c:b_1-\epsilon,b_1+\epsilon)$, are line segments in $\mathscr{D}_0$, and corresponding regions excited by them are $\mathscr{P}_1$ and $\mathscr{P}_{\epsilon}$ from lemma ~\ref{lem:dense_L}. Considering a resonant triad (solid black dots in Fig.~\ref{Fig:proof_sketch}) such that one mode lies in $\mathscr{P}_{\epsilon}$ and the other in $V(a: b_1,b_2)$, implies the third mode must lie in $\mathscr{D}^{R}_{0} \setminus \mathscr{D}_{0}$. \\
Note that the variables $a,b,c$ introduced \textbf{are not the same} as the ones introduced in an earlier section.

\medskip

\noindent\textbf{Step 1: Spreading outside $\mathscr{D}^{R}_0$.}
Consider a vertical line segment $V(a:b_{1},b_{2}) \subset \mathscr{D}_{0}$. If there are two nonzero resonant modes in $V(a:b_{1},b_{2})$, then by Lemma~\ref{lem:dense_L} wave action spreads to the region
\[
\mathscr{P}_{1}
=
\left\{
(\xi,\eta)\in \mathbb{R}^{2}_{>}
\;\middle|\;
b_{1}\leq \eta \leq \xi \leq b_{2} < a
\right\}.
\]
Since $\eta_{\min} \leq b_{1} \leq \eta_{\max}$ and $\eta_{\min} \leq b_{2} \leq \eta_{\max}$, we have
\[
\mathscr{P}_{1} \subset \mathscr{D}^{R}_{2}
\subset \mathbb{R}^{2}_{>} \setminus \mathscr{D}^{R}_{0}.
\]

\medskip

\noindent\textbf{Step 2: Spreading to $\mathscr{D}^{R}_{0}\setminus\mathscr{D}_{0}$.}
Since $\mathscr{D}^{R}_{0}$ is the smallest axis-aligned rectangle enclosing $\mathscr{D}_{0}$, the following property holds.

\medskip

\noindent
(i) For any $\xi_{\min}<\xi<\xi_{\max}$, the vertical segment
$V(\xi:\eta_{\min},\eta_{\max})$ intersects $\mathscr{D}_{0}$ at least at two points
$(\xi,\eta_{-})$ and $(\xi,\eta_{+})$ lying on $\partial\mathscr{D}_{0}$ such that
\[
V(\xi:\eta_{-},\eta_{+}) \subset \mathscr{D}_{0}.
\]
Moreover, for any $\eta_{-}<\eta<\eta_{+}$, there exists $\epsilon>0$ such that
\[
V(\xi:\eta-\epsilon,\eta+\epsilon)\subset V(\xi:\eta_{-},\eta_{+}).
\]
An analogous statement holds for horizontal segments
$H(\eta:\xi_{\min},\xi_{\max})$ with $\eta_{\min}<\eta<\eta_{\max}$.

There are two exceptions to this property. First, the sides of the rectangle, such as
$V(\xi_{\max}:\eta_{\min},\eta_{\max})$ or
$V(\xi_{\min}:\eta_{\min},\eta_{\max})$, are tangent to $\partial\mathscr{D}_{0}$ and therefore do not contain any nontrivial subsegments lying entirely in $\mathscr{D}_{0}$. Second, cut points constitute another exception. Although the theorem remains valid in the presence of cut points, the proof requires minor modifications, which we address later.

\medskip

\noindent
(ii) For any $\xi_{\min}<\xi<\xi_{\max}$, the distance between the points
$(\xi,\eta_{\min})$ and $(\xi,\eta_{-})$ is finite, and
\[
V(\xi:\eta_{\min},\eta_{-}) \subset \mathscr{D}^{R}_{0}\setminus\mathscr{D}_{0}.
\]

\medskip

Now consider a vertical line segment $V(a:b_{1},b_{2}) \subset \mathscr{D}_{0}$ such that
$\{(a,b_{1}),(a,b_{2})\}\subset\partial\mathscr{D}_{0}$, with
\[
\eta_{\min}<b_{1}<\eta_{\max},
\qquad
\eta_{\min}<b_{2}<\eta_{\max}.
\]
By property (ii), the distance between $(a,b_{1})$ and $(a,\eta_{\min})$ is finite; denote it by $\epsilon^{\ast}>0$.
Furthermore, $H(b_{1}:\xi_{\min},\xi_{\max})\cap\mathscr{D}_{0}\neq\emptyset$, more precisely, the intersection consists of points where the horizontal line intersects the domain.
Let $(c,b_{1})\in\mathscr{D}_{0}\setminus\partial\mathscr{D}_{0}$ be such an intersection point.
There exists $\epsilon>0$ such that
\[
V(c:b_{1}-\epsilon,b_{1}+\epsilon)\subset\mathscr{D}_{0}.
\]
If two nonzero resonant modes lie in this segment, then by Lemma~\ref{lem:dense_L} wave action spreads to
\[
\mathscr{P}_{\epsilon}
=
\left\{
(\xi,\eta)\in\mathbb{R}^{2}_{>}
\;\middle|\;
b_{1}-\epsilon\leq\eta<\xi\leq b_{1}+\epsilon<c
\right\}
\subset\mathscr{D}^{R}_{2}.
\]
Now we show that if two modes of a triad lie in $V(a:b_1,b_2)$ and $\mathscr{P}_{\epsilon}$, then the third mode could lie in $\mathscr{D}^{R}_{0}\setminus\mathscr{D}_{0}$.
Consider a resonant triad $\{(a,\eta),(a,\xi_{1}),(\xi_{1},\eta)\}$ with
$(\xi_{1},\eta)\in\mathscr{P}_{\epsilon}$ and $(a,\xi_{1})\in V(a:b_{1},b_{2})$.
Then
\[
b_{1}-\epsilon\leq\eta\leq\xi_{1}\leq b_{1}+\epsilon,
\qquad
b_{1}\leq\xi_{1}\leq b_{2},
\]
which implies
\[
(a,\eta)\in V(a:b_{1}-\epsilon,b_{1}+\epsilon).
\]
  Since $(a,b_1)\in\partial\mathscr{D}_0$ we have
\[
V(a:b_{1}-\epsilon,b_{1})\setminus\{(a,b_1)\}\subset \mathscr{D}^{R}_{0}\setminus\mathscr{D}_{0}.
\]

\medskip
Repeating this argument starting from the horizontal segment
$H(b_{1}-\epsilon:\xi_{\min},\xi_{\max})$, we obtain an $\epsilon_{1}>0$ such that wave action spreads to
\[
\mathscr{P}_{\epsilon_{1}}
=
\left\{
(\xi,\eta)\in\mathbb{R}^{2}_{>}
\;\middle|\;
b_{1}-\epsilon-\epsilon_{1}\leq\eta\leq\xi\leq b_{1}-\epsilon+\epsilon_{1}\leq c_{2}
\right\}
\subset\mathscr{D}^{R}_{2},
\]
and consequently to
\[
V(a:b_{1}-\epsilon-\epsilon_{1},b_{1}-\epsilon)
\subset \mathscr{D}^{R}_{0}\setminus\mathscr{D}_{0}.
\]

Iterating this procedure, we obtain a decreasing sequence
\[
b_{1}>b_{1}-\epsilon>b_{1}-\epsilon-\epsilon_{1}>\cdots
\]
until $b_{1}^{\ast}\leq\eta_{\min}$. This occurs in finitely many steps since $\epsilon^{\ast}$ is finite. Hence, the wave action is non-zero for modes in $V(a:\eta_{\min},b_{1})\subset \mathscr{D}^{R}_{0}\setminus\mathscr{D}_{0}.$

An analogous argument starting from $(a,b_{2})$ results in the wave action spreading to $
V(a:b_{2},\eta_{\max})\subset \mathscr{D}^{R}_{0}\setminus\mathscr{D}_{0}$.

Therefore, after finitely many steps, wave action is nonzero on the entire segment
$V(a:\eta_{\min},\eta_{\max})$.

Choosing two resonant modes on this segment and applying Lemma~\ref{lem:dense_L}, we obtain
\[
\mathscr{P}
=
\left\{
(\xi,\eta)\in\mathbb{R}^{2}_{>}
\;\middle|\;
\eta_{\min}\leq\eta\leq\xi\leq\eta_{\max}<a
\right\}=\mathscr{D}^{R}_{2}.
\]
Therefore, the wave action spreads to all $\mathscr{D}^{R}_{2}$.
By symmetry, the same argument but starting with horizontal line segments results in non-zero wave action in $\mathscr{D}^{R}_{1}$. Now, once wave action spreads to the regions $\mathscr{D}^{R}_0$, $\mathscr{D}^{R}_1$, and $\mathscr{D}^{R}_2$ from Lemma~\ref{lem:D012} we know that wave action fills the region $\mathscr{D}^{R}_{0}\cup\mathscr{D}^{R}_{1}\cup\mathscr{D}^{R}_{2}$ and from the Theorem ~\ref{th:fsuppsquare} the wave action does not spread to new regions. Thus $t>0$ 
\[
\operatorname{supp} n(\cdot,t)
=
\mathscr{D}^{R}_{0}\cup\mathscr{D}^{R}_{1}\cup\mathscr{D}^{R}_{2}.
\]
\end{proof}
%%%%%%%%%%%%
%%%%%%%%%%%%%%%%
 \begin{figure}[!]
\centering
\includegraphics[width=\linewidth]{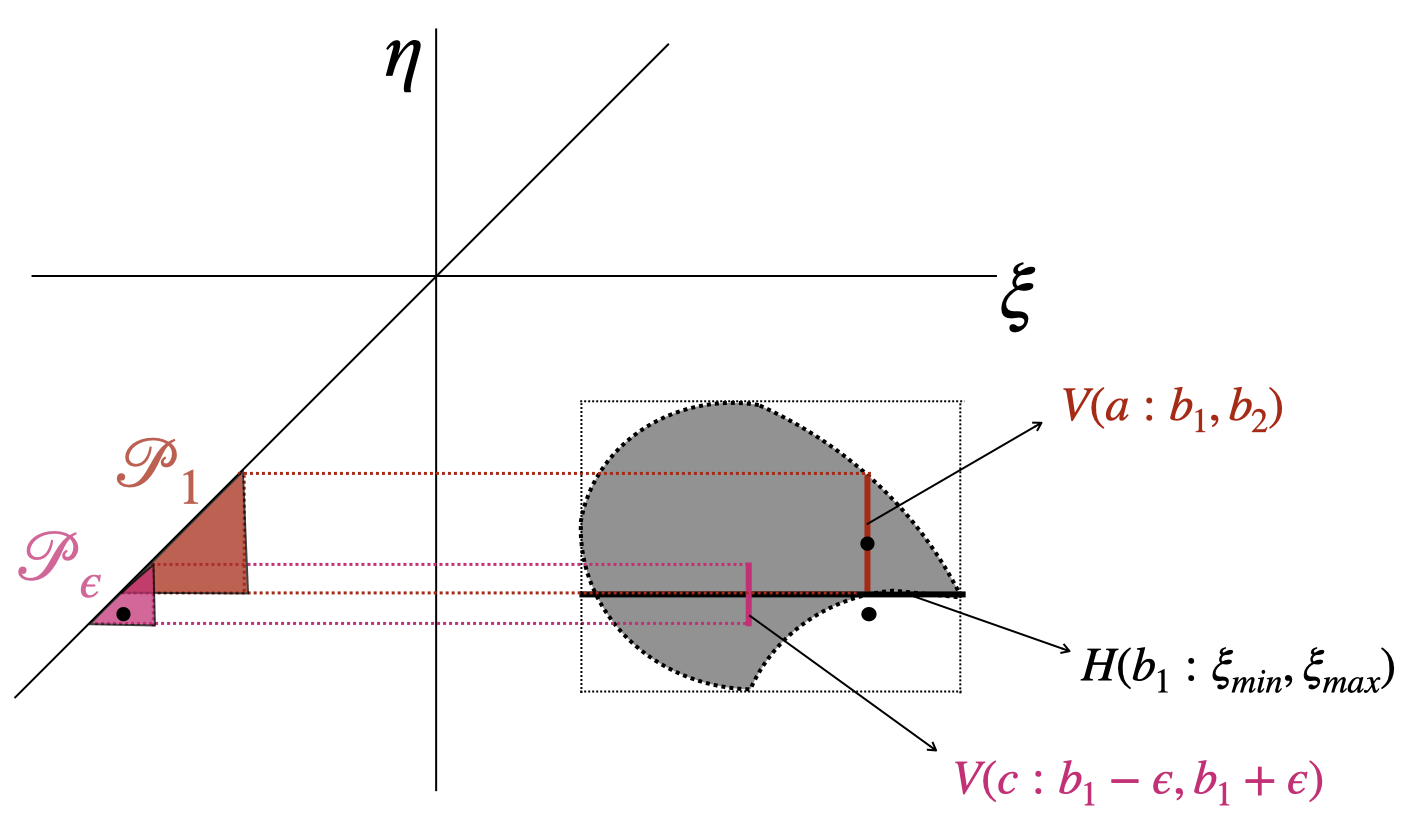}
\caption{ Schematic of a compact domain $\mathscr{D}_0$ (shaded gray), its minimal enclosing rectangle $\mathscr{D}_0^{R}$. If two modes of a resonant triad lie on the line $V(a:b_1,b_2)$ the third mode lies in $\mathscr{P}_1$. And similarly if two modes lie on the line $V(c:b_1-\epsilon,b_1+\epsilon)$ the third mode lies in $\mathscr{P}_{\epsilon}$. Where the point $(c,b_1)\in \mathscr{D}_0$ is a point of intersection between lines $H(b_1:\xi_{min},\xi_{\max})$ and $V(c:b_1-\epsilon,b_1+\epsilon)$.  The black dots show a representative resonance triad, note the triad is formed by a mode in $\mathscr{P}_{\epsilon}$ and $V(a:b_1,b_2)$ such that the third mode is in $\mathscr{D}^{R}_0\setminus\mathscr{D}_0$.
}
\label{Fig:proof_sketch}
\end{figure}
%%%%%%%%%%%%%
\begin{figure}[!]
\centering
\includegraphics[width=\linewidth]{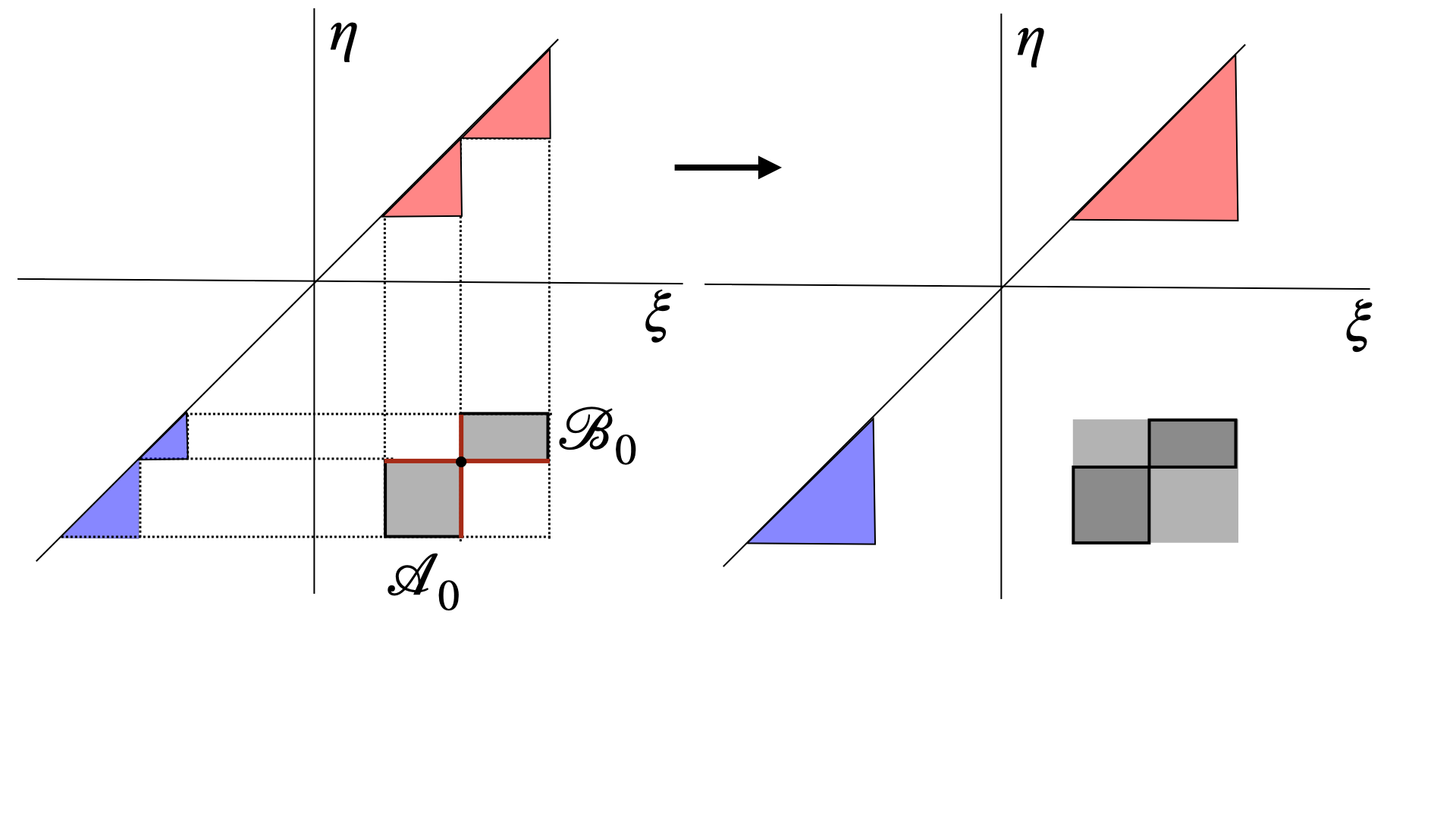}
\caption{A schematic illustrating how the wave action spreads to the final support when a cut-point is involved.}
\label{Fig:cutpoint}
\end{figure}

The above result also extends to domains containing a cut-point. For instance, in Fig.~\ref{Fig:cutpoint}, we consider a domain $\mathscr{D}_0$ with a cut-point. In the absence of this cut-point, $\mathscr{D}_0$ can be decomposed into two disjoint subdomains, $\mathscr{A}_0$ and $\mathscr{B}_0$. The rectangles shown correspond to their respective minimal enclosing rectangles.

Initially, the dynamics generate resonant interactions along the triangular regions and within the enclosing rectangles associated with $\mathscr{A}_0$ and $\mathscr{B}_0$ separately. However, the presence of the shared cut-point enables coupling between these subdomains: wave action spreads along the horizontal and vertical lines passing through the cut-point. As a result, and by virtue of Lemma~\ref{lem:dense_L}, the long-time support expands to the minimal enclosing rectangle corresponding to the union $\mathscr{A}_0^{R} \cup \mathscr{B}_0^{R}$.\\
\newpage
Now we extend the result to fully disconnected regions, see Fig.~\ref{Fig:disconnsup}.

\begin{theorem}
Let $n(\xi,\eta,t)$ be a solution of a wave kinetic equation in $\mathbb{R}^2_{>}$. 
Assume
\[
\operatorname{supp} n(\cdot,0)=\mathcal D_0=\mathcal A_0\cup\mathcal B_0,
\]
where $\mathcal A_0$ and $\mathcal B_0$ are compact subsets of $\mathbb{R}^2_{>}$. Suppose
\[
\pi_{\xi}(\mathcal A_0)\cap\pi_{\xi}(\mathcal B_0)=\varnothing,
\qquad
\pi_{\eta}(\mathcal A_0)\cap\pi_{\eta}(\mathcal B_0)=\varnothing ,
\]
where $\pi_{\xi}(x,y)=x$ and $\pi_{\eta}(x,y)=y$ are the horizontal and the vertical projections respectively.

Let $\mathcal A_0^R$ and $\mathcal B_0^R$ denote their smallest enclosing rectangles and 
$\mathscr{A}^{R}_1$, $\mathscr{A}^{R}_2$, $\mathscr{B}^{R}_1$, and $\mathscr{B}^{R}_2$ 
the corresponding triangular regions (see Definition~\ref{def:supports}).

Then for $t>0$,
\[
\operatorname{supp} n(\cdot,t)
=
(\mathscr{A}^{R}_{0}\cup \mathscr{A}^{R}_{1}\cup\mathscr{A}^{R}_{2})
\cup
(\mathscr{B}^{R}_{0}\cup \mathscr{B}^{R}_{1}\cup\mathscr{B}^{R}_{2}).
\]
Moreover, the evolution of $n(\xi,\eta,t)$ is decoupled in the $\mathcal A$ and $\mathcal B$ regions; that is, there are no resonant interactions coupling the two regions.
\end{theorem}

\begin{proof}
By definition of the enclosing rectangles,
\[
\pi_{\xi}(\mathcal A_0)=\pi_{\xi}(\mathcal A^{R}_0),
\quad
\pi_{\eta}(\mathcal A_0)=\pi_{\eta}(\mathcal A^{R}_0),
\]
\[
\pi_{\xi}(\mathcal B_0)=\pi_{\xi}(\mathcal B^{R}_0),
\quad
\pi_{\eta}(\mathcal B_0)=\pi_{\eta}(\mathcal B^{R}_0).
\]

Write
\[
\mathscr{A}^{R}_0
=
\{(\xi,\eta)\in\mathbb{R}^2_{>}
\mid
\xi_{\min} \leq \xi \leq \xi_{\max},\;
\eta_{\min} \leq \eta \leq \eta_{\max}\},
\]
\[
\mathscr{B}^{R}_0
=
\{(\xi,\eta)\in\mathbb{R}^2_{>}
\mid
\tilde\xi_{\min} \leq \xi \leq \tilde\xi_{\max},\;
\tilde\eta_{\min} \leq \eta \leq \tilde\eta_{\max}\}.
\]

Since the projections are disjoint in both coordinate directions, we may assume (without loss of generality) the ordering
\[
\xi_{\min}<\xi_{\max}<\tilde\xi_{\min}<\tilde\xi_{\max},
\]
\[
\eta_{\min}<\eta_{\max}<\tilde\eta_{\min}<\tilde\eta_{\max}.
\]

By Lemma~\ref{lem:D012}, if two modes of a resonant triad
$\{(\xi,\eta),(\xi,\xi_1),(\xi_1,\eta)\}$ belong to
$\mathscr{A}^{R}_{0}\cup \mathscr{A}^{R}_{1}\cup\mathscr{A}^{R}_{2}$,
then the third mode also belongs to this same region.
An analogous statement holds for
$\mathscr{B}^{R}_{0}\cup \mathscr{B}^{R}_{1}\cup\mathscr{B}^{R}_{2}$.

We now show that no resonant triad can couple the $\mathcal A$ and $\mathcal B$ regions.
Any potential coupling would require, for example,
\[
(\xi_1,\eta)\in \mathscr{A}^{R}_{0}\cup \mathscr{A}^{R}_{1}\cup\mathscr{A}^{R}_{2},
\qquad
(\xi,\xi_1)\in \mathscr{B}^{R}_{0}\cup \mathscr{B}^{R}_{1}\cup\mathscr{B}^{R}_{2}.
\]
However, this is impossible due to the strict separation
\[
\xi_{\max}<\tilde\xi_{\min},
\qquad
\eta_{\max}<\tilde\eta_{\min}.
\]
Indeed, any point in an $\mathcal A$-region has both $\xi$- and $\eta$-coordinates strictly smaller than those of any point in a $\mathcal B$-region. Therefore, the resonance relations cannot be simultaneously satisfied with modes drawn from both regions. Hence, resonant interactions between $\mathcal A$ and $\mathcal B$ are ruled out.

Finally, by Theorem~\ref{th:compD0}, wave action starting from $\mathcal A_0$
spreads to
$\mathscr{A}^{R}_{0}\cup \mathscr{A}^{R}_{1}\cup\mathscr{A}^{R}_{2}$.
Similarly, wave action starting from $\mathcal B_0$
spreads to
$\mathscr{B}^{R}_{0}\cup \mathscr{B}^{R}_{1}\cup\mathscr{B}^{R}_{2}$.

Since the dynamics in the two regions are decoupled, it follows that for $t>0$
\[
\operatorname{supp} n(\cdot,t)
=
(\mathscr{A}^{R}_{0}\cup \mathscr{A}^{R}_{1}\cup\mathscr{A}^{R}_{2})
\cup
(\mathscr{B}^{R}_{0}\cup \mathscr{B}^{R}_{1}\cup\mathscr{B}^{R}_{2}).
\]
\end{proof}
\begin{figure}[!]
\centering
\includegraphics[width=\linewidth]{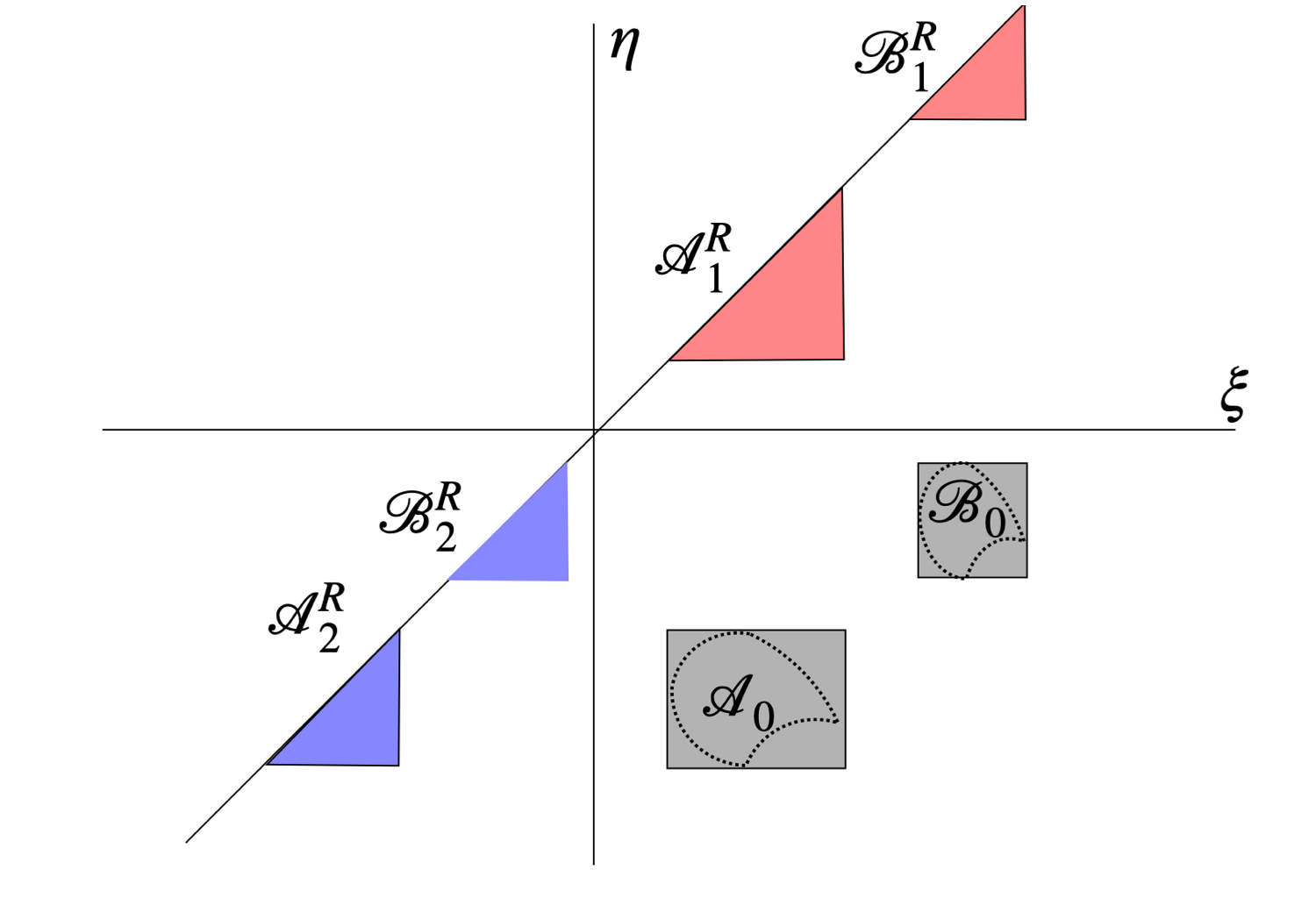}
\caption{A schematic illustrating how the wave action spreads to the final support when the sets are disjointed without any overlapping projections.}
\label{Fig:disconnsup}
\end{figure}
Finally, even when the initial compact support is disconnected, but its constituent components have overlapping projections along either the $\xi$- or $\eta$-direction, the support remains finite for all $t>0$. For instance, consider the schematic in Fig.~\ref{Fig:overlap_proj}, where the initial domain consists of two disconnected compact sets, $\mathscr{A}_0 \cup \mathscr{B}_0$. These sets have overlapping projections along the $\eta$-direction, i.e.,:
\[
\pi_{\eta}(\mathcal A_0)\cap\pi_{\eta}(\mathcal B_0)\neq\varnothing\]
The resulting long-time support remains bounded, as illustrated by the shaded region in Fig.~\ref{Fig:overlap_proj}. However, the structure of the support must be redefined due to the emergence of additional secondary resonances.
First, the rectangles enclosing $\mathscr{A}_0$ and $\mathscr{B}_0$ are no longer minimal enclosing rectangles in the previous sense (the dashed black lines demarcate the minimal enclosing rectangles from the redefined regions); instead, they extend equally in the $\eta$-direction, satisfying
\[
\pi_{\eta}(\mathscr{A}^{R}_0)= \pi_{\eta}(\mathscr{B}^{R}_0)=\pi_{\eta}(\mathscr{A}_0)\cup\pi_{\eta}(\mathscr{B}_0).
\]
Second, coupling between the rectangles $\mathscr{A}^{R}_0$ and $\mathscr{B}^{R}_0$ generates an additional rectangular region, denoted by $\mathscr{C}^{R}$. For example, the black dots in Fig.~\ref{Fig:overlap_proj} represent a resonant triad that facilitates the transfer of wave action into the region $\mathscr{C}^{R}$. In the next section a numerical simulation showing the same is plotted.
%%%%%%%%%%%%%%

%%%%%%%%%%%%%%
\begin{figure}[!]
\centering
\includegraphics[width=\linewidth]{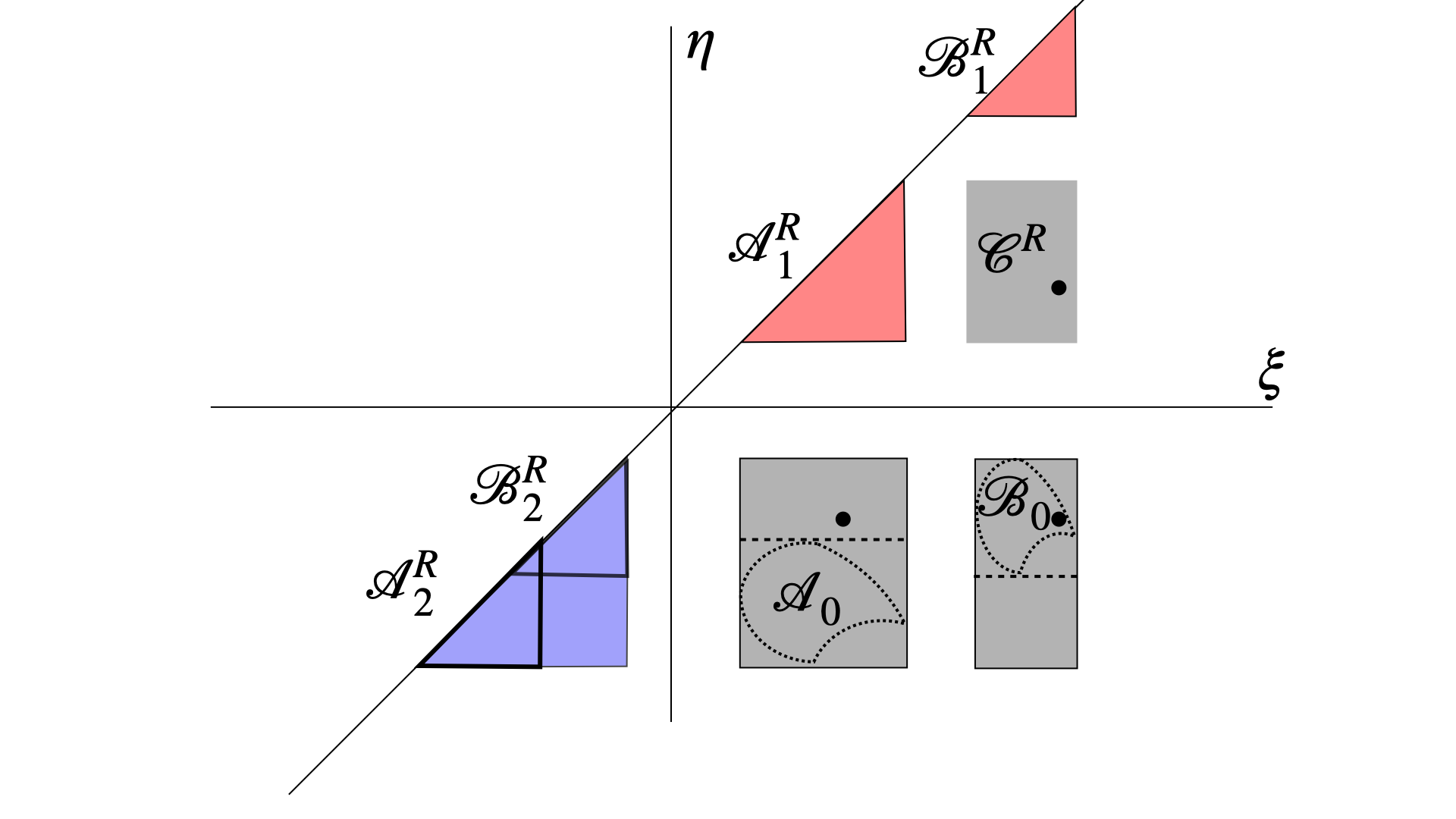}
\caption{A schematic illustrating how the wave action spreads to the final support when the initial disconnected compact domain has no overlapping projections amongst disconnected regions.}
\label{Fig:overlap_proj}
\end{figure}

\section{Additional numerical results}\label{methods:add_sim}
\begin{itemize}
    \item 
In Fig.~\ref{Fig:evolRunA} we plot the temporal evolution of $n(\xi,\eta)$ for Run A from the maintext.
\item 
In Fig.~\ref{Fig:fitData}, we reconstruct a fitted spectrum $\hat{n}(\xi,\eta)$ using data along a single horizontal line within $\mathscr{D}_0^{R}$, as described in the Methods section. An arbitrary choice of $\eta_0$ is sufficient, and the procedure is valid for any $-b < \eta_0 < -a$.
We find excellent agreement between the reconstructed fit and the simulation, with the relative error remaining at most $\mathcal{O}(10^{-12})$ for Run A and $\mathcal{O}(10^{-11})$ for Run B. For Run A, the fit is constructed using the wave-action spectrum at $t/\tau = 3 \times 10^{5}$, while for Run B, we use data at $t/\tau = 1.1 \times 10^{5}$.
\item 
For Run A spectrum, we take again a horizontal line within $\mathscr{D}_0^{R}$ and fit $n^{-1}(\xi,\eta_0)$ with a polynomial which has odd powers in $\xi$ or equivalently we take:
\[
F(x)=\sum^{i=N}_{i=0}c_ix^{2i-1}
\]
Using the fit with with $N=10$, we construct a $\hat{n}_p(\xi,\eta)=1/(F(\xi)-F(\eta))$. In Fig.~\ref{Fig:fitPoly} we plot the fit, and note that a maximum relative error of $\sim\mathcal{O}(10^{-2})$. Thus, even though it is attractive to have a low-dimensional polynomial approximation in place of the original functional fit, it comes at a dramatic loss of accuracy.
\item 
In Fig~\ref{Fig:evolRunC} we show numerical results for Run C such that the support of initial condition is given by $\mathscr{D}_0=\mathscr{A}_0\cup\mathscr{B}_0$, where:
\[
\mathscr{A}_0=\{(\xi,\eta):(\xi-2)^2 + (\eta+3)^2 \leq 0.75^2\},~~~
\mathscr{B}_0=\{(\xi,\eta):(\xi-4)^2 + (\eta+2)^2 \leq 0.75^2\}.\\
\]
It's clear that $\pi_{\xi}(\mathcal A_0)\cap\pi_{\xi}(\mathcal B_0)=\varnothing$ and $\pi_{\eta}(\mathcal A_0)\cap\pi_{\xi}(\mathcal B_0)\neq\varnothing$. This precisely corresponds to the scenario described by Fig.~\ref{Fig:overlap_proj}. For initial conditions we choose 
\[
n(\xi,\eta,0)=n_A+n_B,
\]
where,
\[
n_A(\xi,\eta)= \exp\!\left(-\frac{(\xi-2)^2}{0.75^2}-\frac{(\eta+3)^2}{0.75^2}\right), \quad (\xi,\eta)\in \mathscr{A}_0,~~~~
n_B(\xi,\eta)= \exp\!\left(-\frac{(\xi-4)^2}{0.75^2}-\frac{(\eta+2)^2}{0.75^2}\right), \quad (\xi,\eta)\in \mathscr{B}_0.\\
\]
In Fig.~\ref{Fig:thermalizationRunC} we plot $\mathscr{N}_{\infty}$ and $\mathcal{S}$ versus $t/\tau_{nl}$. We note that the system thermalizes like Runs A and B.
\end{itemize}
\begin{figure}[!]
\centering
\includegraphics[width=\linewidth]{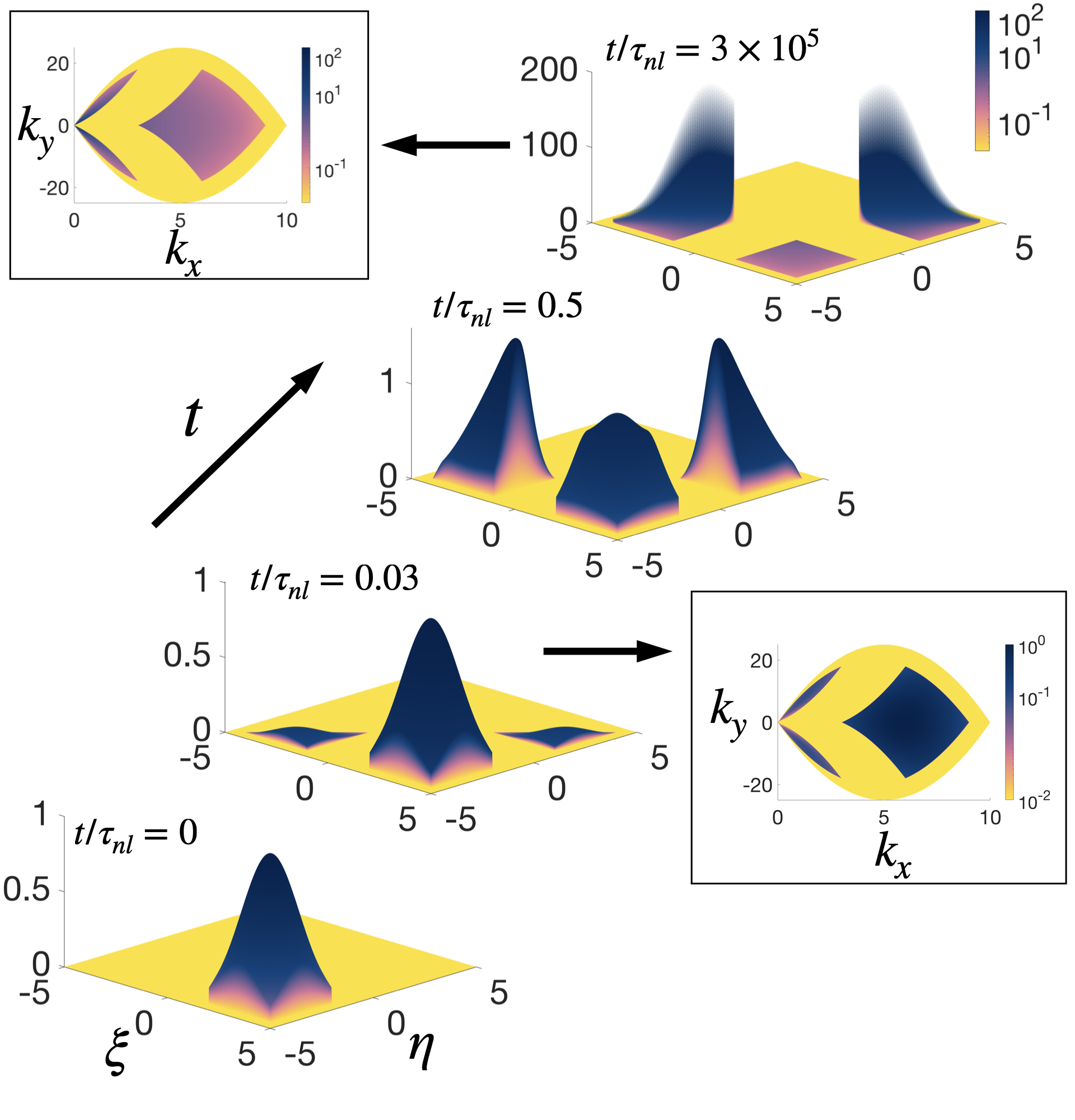}
\caption{Temporal evolution of $n(\xi,\eta)$ for Run A. In the insets we plot the spectrum in the $(k_x,k_y)$ co-ordinates.}
\label{Fig:evolRunA}
\end{figure}
\begin{figure}[!]
\centering
\includegraphics[width=\linewidth]{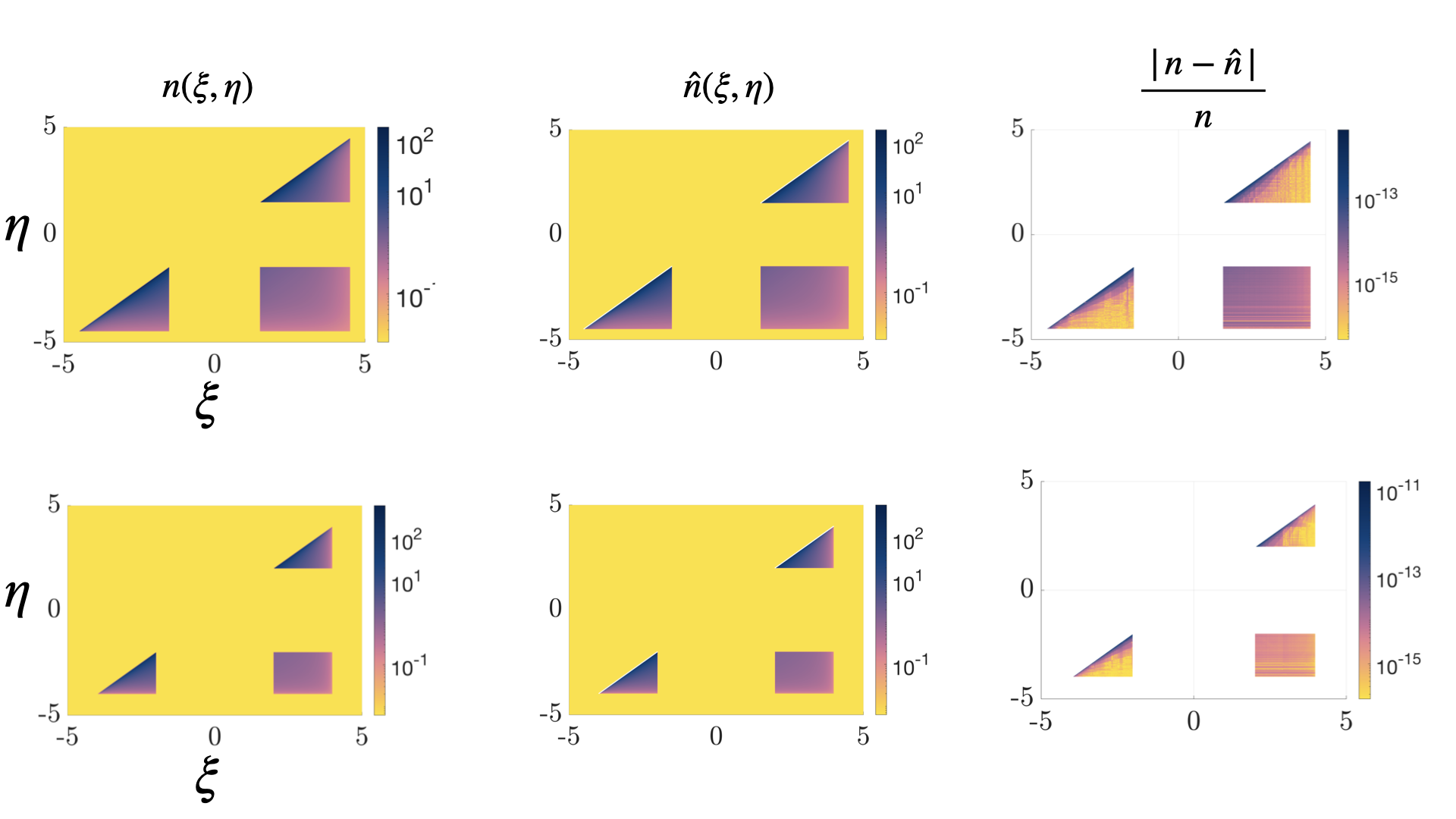}
\caption{Columns (I–III) show, respectively, the wave-action spectrum obtained from simulations at long times $(\tau/\tau_{nl} \sim \mathcal{O}(10^{5})$, the fitted spectrum $\hat{n}(\xi,\eta)$, and the relative error between the simulation and the fit. Rows (I–II) correspond to Run A and Run B, respectively.}
\label{Fig:fitData}
\end{figure}
\begin{figure}[!]
\centering
\includegraphics[width=\linewidth]{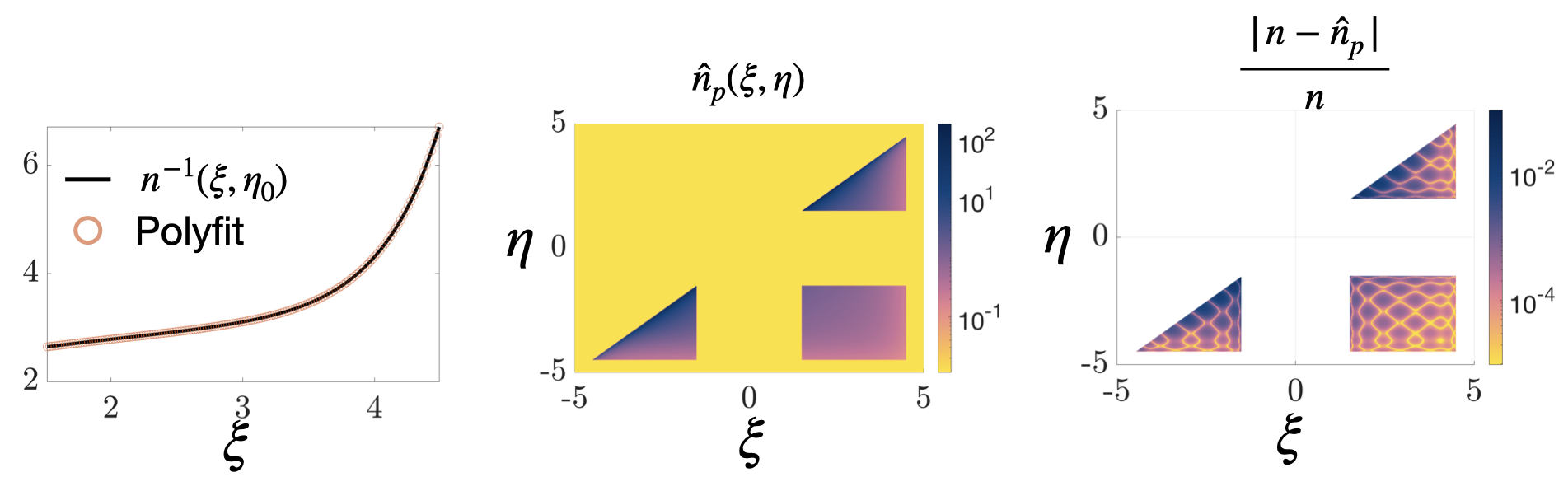}
\caption{For Run A Columns (I–III) show, respectively, the horizontal slice  $n^{-1}(\xi,\eta_0)$ versus $\xi$ and the corresponding fit, wave-action spectrum $\hat{n}_p(\xi,\eta)$ obtained from a polynomial fit,  and the relative error between the simulation and the fit. }
\label{Fig:fitPoly}
\end{figure}
\begin{figure}[!]
\centering
\includegraphics[width=\linewidth]{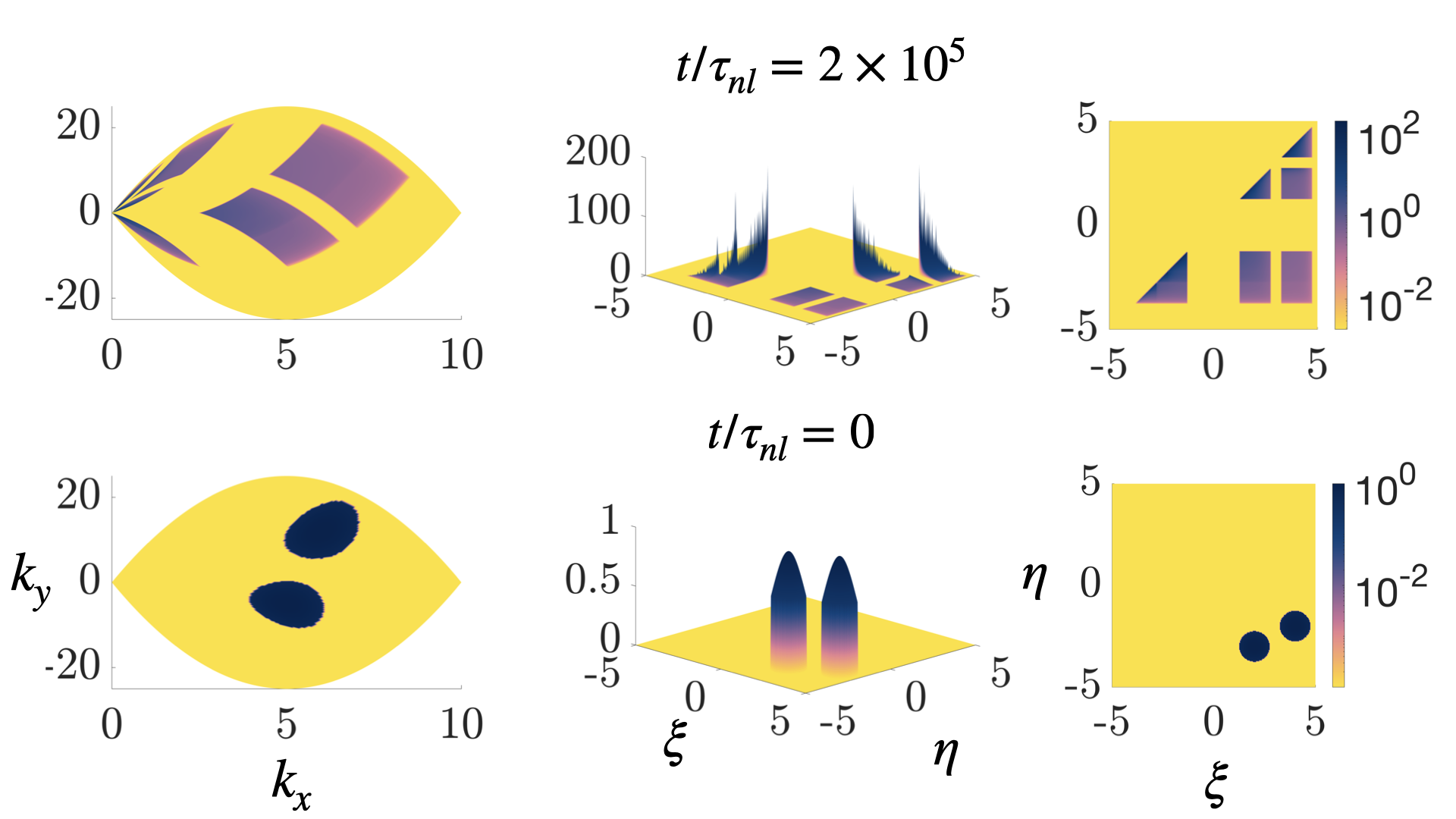}
\caption{Spectrum evolution in Run C. Left column: 2D color plot of $n(k_x,k_y)$. Center and right columns:
spectrum in the 
$(\xi,\eta)$ space -- the landscape and the 2D plots respectively.}
\label{Fig:evolRunC}
\end{figure}
\begin{figure}[!]
\centering
\includegraphics[width=\linewidth]{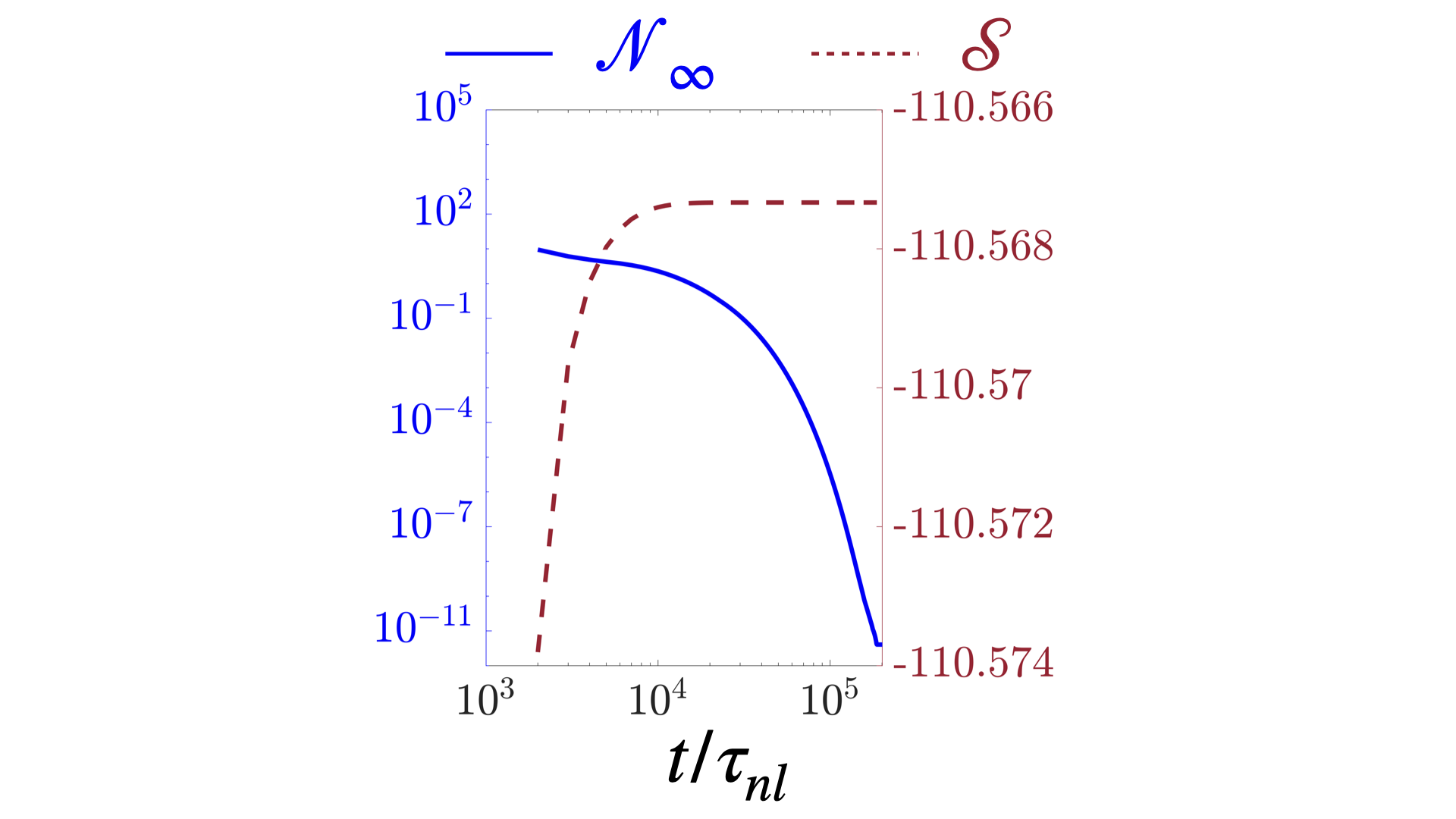}
\caption{Evolution of bracket ${\cal N}_\infty$ (solid line) and entropy $\mathcal{S}$ (dashed line).}
\label{Fig:thermalizationRunC}
\end{figure}

\end{document}